\documentclass{amsart}
\usepackage{graphicx}
\usepackage{amsmath,amssymb}
\usepackage{eqparbox}
\usepackage{hyperref}
\usepackage{enumerate}
\usepackage[shortlabels]{enumitem}
\usepackage{tikz}
\usepackage{todonotes}
\usepackage{comment}
\usepackage[T2A]{fontenc}
\usepackage[utf8]{inputenc}
\usepackage{float}
\usepackage[pagewise]{lineno}

\def\A{\mathcal{A}}
\def\B{\mathcal{B}}
\def\E{\mathcal{E}}
\def\G{\mathcal{G}}
\def\LL{\mathcal{L}}
\def\N{\mathbb{N}}
\def\R{\mathcal{R}}
\def\Z{\mathbb{Z}}
\def\ff{\textbf{f}}
\def\ww{\textbf{w}}

\newcommand{\bwt}[2]{{\rm bwt}_{#1} \left( #2 \right)}
\newcommand{\ebwt}[2]{{\rm ebwt}_{#1} \left( #2 \right)}
\newcommand{\card}[1]{{\rm Card}\left( #1 \right)}
\newcommand*{\eqmathbox}[2][M]{\eqmakebox[#1]{$\displaystyle#2$}}

\theoremstyle{plain}
\newtheorem{theorem}{Theorem}[section]
\newtheorem{proposition}[theorem]{Proposition}
\newtheorem{lemma}[theorem]{Lemma}
\newtheorem{corollary}[theorem]{Corollary}

\newtheorem{example}[theorem]{Example}

\theoremstyle{remark}
\newtheorem{remark}[theorem]{Remark}

\title{Extended branching Rauzy induction}

\author[Francesco Dolce]{Francesco Dolce}
\address{Department of Applied Mathematics, Faculty of Information Technology, Czech Technical University in Prague, Prague, Czech Republic}
\email{dolcefra@cvut.cz}

\author[Christian B. Hughes]{Christian B. Hughes}
\address{Faculty of Information Technology, Czech Technical University in Prague, Prague, Czech Republic}
\email{hughechr@cvut.cz}

\begin{document}

\begin{abstract}
Branching Rauzy induction is a two-sided form of Rauzy induction that acts on regular interval exchange transformations (IETs). We introduce an extended form of branching Rauzy induction that applies to arbitrary standard IETs, including non-minimal ones.
The procedure generalizes the branching Rauzy method with two induction steps, merging and splitting, to handle equal-length cuts and invariant components respectively.
As an application, we show, via a stepwise morphic argument, that all return words in the language of an arbitrary IET cluster in the Burrows-Wheeler sense.
\end{abstract}


\maketitle

\section{Introduction}
\label{sec:intro}

Interval exchange transformations, or IETs, were initially introduced by Oseledec in~\cite{Oseledec66}, following an earlier idea of Arnold~\cite{Arnold63}.
They have long been a subject of interest in the area of dynamical systems due to their wide applicability, ranging from symbolic dynamics to ergodic theory.
Such transformations are defined as bijections of an interval which are piecewise translations.
Each such translation is specified by a permutation on a finite set of sub-intervals and a length vector, which together determine the order and lengths of the exchanged intervals.
IETs preserve Lebesgue measure and have topological entropy $0$, but can display rich behavior.
The simplest non-trivial IETs are equivalent to rotations on a circle.

Given the partition of the interval into exchanged sub-intervals, we code the orbit of a point in the interval by recording, at each iterate of the IET, the label of the sub-interval it lands in.
This produces infinite words, called trajectories, over the alphabet associated to the sub-intervals.
The set of all finite factors of these trajectories is called the language of the IET.
We analyze these languages using tools from combinatorics on words and formal language theory.

Such transformations admit geometric induction via Rauzy induction~\cite{Rauzy79} (also called Rauzy-Veech induction~\cite{Veech78}), and other induction schemes, including induction on unions of (possibly disjoint) sub-intervals (see, e.g.,~\cite{viana2006ergodic}). The induction arguments used in the literature often employ some assumption of \emph{regularity} of the transformation.
The most common one is the \emph{Keane condition}, also known as the \emph{infinite disjoint orbit condition} (or \emph{i.d.o.c.}).
The presence of periodic components and the reducibility of the permutation often disrupt the ability to induce on a given sub-interval in a manner which respects the underlying symbolic coding of the system.
In~\cite{branching} a two-sided version of Rauzy induction was introduced.

In this paper, we develop an \emph{extended branching Rauzy induction} that allows induction on a cylinder for \emph{all} standard IETs, without any assumption of regularity (nor minimality) of the transformation.
This procedure augments the typical left/right induction steps of the original branching Rauzy induction with two induction steps: \emph{merging}, which resolves equal-length right/left cuts by collapsing the two sub-intervals involved in the cut into one interval, and \emph{splitting}, which separates invariant components (minimal or periodic) of the IET.

This construction has two principal consequences.
Firstly, for every cylinder $I_w$, with $w$ in the language of the IET $T$, there exists a \emph{finite} induction chain of extended branching Rauzy induction steps $\chi$ such that $\chi (T)$ is the first-return map of $T$ to $I_w$.
Secondly, each step is associated to a morphism on words, so the entire composed chain of step-associated morphism yields a morphism which "transports" the induced coding back to that of the original system.
In particular, the procedure yields constructive descriptions of the set of return words in the language generated by an IET.
These additional morphisms clarify how geometric cuts manifest in the symbolic coding.

We phrase much of the language-theoretic part of this paper in terms of extension graphs, which are bipartite graphs that describe how factors can be extended on the left and right by letters of the alphabet in the underlying language.
A language such that all its extension graphs are trees is called \emph{dendric} (see, e.g.~\cite{acyclicconnectedtree} or~\cite{DolcePerrin19,GheeraertLejeuneLeroy22} where the notion of dendricity is applied to shift spaces.).

The Burrows-Wheeler transform, introduced by Burrows and Wheeler in~\cite{BurrowsWheeler94}, was originally investigated in the context of data compression and has since been studied from a symbolic and combinatorial point of view (see, e.g.,~\cite{DolceHughes2025,FerencziZamboni13,ferenczi2025clustering}).
We use \emph{ordered alsinicity} (the notion of an extension graph being a forest) to obtain an application to Burrows-Wheeler clustering in symbolic codings of IETs.
In particular, we generalize a result in~\cite{DolceHughes2025} by showing that, for any standard IET with permutation $\pi$, every return word in its language is $\pi$-clustering in the Burrows-Wheeler sense.
This settles a question originally raised in ~\cite{lapointe2021perfectly}, in which the author asked whether this result holds when $\pi$ is the symmetric permutation.
Note that this result was recently shown independently in~\cite{ferenczi2025clustering} using different arguments.

The paper is organized as follows.
In Section~\ref{sec:cow} we recall basic notions and notation of combinatorics on words.
These notions will be used both to connect the studied dynamical systems to their symbolic coding, and to formulate questions concerning their link to the Burrows-Wheeler clustering property in combinatorial terms.
In Section~$\ref{sec:iets}$, we state some classical results concerning Interval Exchange Transformations.
In Section~\ref{sec:Rauzy}, we define our extended induction steps and show that they behave as expected.
Section~\ref{sec:forest} is devoted to the link between ordered alsinicity and clustering; there we provide the necessary background for the proof of the Theorem~\ref{thm:picluster} in Section~\ref{sec:piclusteringsection}, which serves as a key application of this paper.

\section{Combinatorics on Words}
\label{sec:cow}

For all undefined terms, we refer the reader to~\cite{Lothaire2}.

\subsection{Words}
\label{sec:words}

An \emph{ordered alphabet} $\A = \{ a_1 < a_2 < \ldots < a_k \}$ is a set of symbols called \emph{letters} together with a total order of its elements.
The number of distinct letters in $\A$ is called the \emph{cardinality} of the alphabet $\A$ and is denoted by $\card{\A}$.
In this paper, we only consider finite alphabets.

The set of \emph{finite words} $\A^*$ over $\A$ is the free monoid over $\A$ with neutral element the \emph{empty word}, denoted by $\varepsilon$.
The product of two words $u,v \in \A^*$ is given by their composition $u \cdot v$, also written $uv$.
We denote the free semigroup over $\A$ by $\A^+$.
That is, $\A^+ = \A^* \setminus \{ \varepsilon \}$.
When it is possible to write $w = pus$, with $p,u,s \in \A^*$, we call $p$ (resp., $s$, $u$) a \emph{prefix} (resp., \emph{suffix}, \emph{factor}) of $w$.

The lexicographic order naturally extends the order on $\A$ to $\A^*$.
In particular, given two distinct words $u = u_0 \cdots u_m$ and $v = v_0 \cdots v_n$, in $\A^*$, we say that $u < v$ if and only if either $u$ is a prefix of $v$, or the inequality between letters $u_i < v_i$ holds at the smallest index $i$ at which $u$ and $v$ differ.

For a given word $w = w_0 w_1 \cdots w_{n-1}$, where each $w_i \in \A$, we denote its length $n$ by $|w|$, and the number of times $u \in \A^+$ appears as a factor of $w$ by $|w|_u$.
The Parikh vector of a word $w \in \A^*$ is the vector $\Psi_{\A}(w) \in \N^k$ defined as $(\Psi_{\A}(w))_{a} = |w|_a$ for each letter $a \in \A$.
A word $w \in \A^*$ is \emph{pangrammatic} (in $\A$) if $\Psi_{\A}(w)$ is a positive vector, i.e., if $|w|_a > 0$ for every $a \in \A$.
These vectors can be generalized to multisets of words over the same ordered alphabet in a natural way (see Example~\ref{ex:ebwt} below).

We call a contiguous sequence of identical letters in a word a \emph{run}.

The $k^\text{th}$ \emph{power} of a word $u$, for $k \in \N$, is defined as $u^0 = \varepsilon$ and $u^k = u^{k-1}u$ for $k>0$. 
A word $w$ is said to be \emph{primitive} if it is not the integer power of another word, i.e., if $w = u^k$ implies $k=1$.
Two words $w, w'$ are \emph{conjugate} if $w = uv$ and $w'=vu$ for some $u,v \in \A^*$.
If a word $w$ is primitive, then it has exactly $|w|$ distinct conjugates.
A \emph{Lyndon word} is a primitive word that is minimal for the lexicographic order among its conjugates.
The class of conjugates of a word $w$ is also called a \emph{circular word} or \emph{necklace} and is denoted by $[w]$.

A (right) \emph{infinite word} (resp., \emph{bi-infinite word}) over $\A$ is a sequence
$\ww = w_0 w_1 w_2 \cdots$
(resp., $\ww = \cdots a_{-1} a_0 a_1 \cdots$), with $w_i \in \A$ for all $i \in \N$ (resp., for all
$i \in \Z)$.
Similarly to finite words, we can define the set of infinite words $\A^\infty$ and extend the notions of prefix, suffix and factor to $\A^\omega = \A^* \cup \A^\infty$ in a natural way.
An infinite word $\ww \in \A^\infty$
is \emph{eventually periodic} if $\ww = u v^\omega = u v v v \cdots$ for some $u,v \in \A^*$.
We say that $\ww$ is \emph{purely periodic} when $u = \varepsilon$.
An infinite word that is not eventually periodic is called \emph{aperiodic}.

\subsection{Languages}
\label{sec:languages}

By \emph{language}, we mean a factorial and bi-extendable set $\LL \subset \A^*$, i.e., a subset of the free monoid such that for every $w \in \LL$, we have $v, aw, wb \in \LL$ for every factor $v$ of $w$ and for some $a,b \in \A$.
The language of an infinite word $\ww$ is the set $\LL(\ww)$ of all its factors, while the language of a finite word $w$ is defined as $\LL(w^\omega)$.
Note that, in this context, the language of a finite word contains infinitely many elements.

A language $\LL$ is \emph{recurrent} if, for every $v \in \LL$, $vuv \in \LL$ for some word $u \in \A^*$.
It is \emph{uniformly recurrent} if for every $v \in \LL$, there exists $N=N(v) \in \N$ such that $v$ appears as a factor of every element of length $N$ in $\LL$.

The set $\R^\text{left}_{\LL}(w)$ of \emph{left return words} to $w$ in $\LL \subset \A^*$
is the set of words $u$ such that $uw \in \LL$ has exactly two occurrences of $w$ as factor: as a prefix and as a suffix.
Formally
$$
    \R_{\LL}^\text{left}(w) =
    \{ u \in \A^* \; : \; uw \in (\LL \cap w \A^*) \setminus \A^+ w \A^+ \}.
$$
Similarly, we can define the set of \emph{right return words} to $w$ in $\LL$ by
$$
    \R_{\LL}^{\text{right}}(w) =
    \{ v \in \A^* \; : \; wv \in (\LL \cap \A^*w) \setminus w\A^+ w \}.
$$
The two sets are related, since for every $w \in \LL$ we have
$$
    w \R_{\LL}^\text{left}(w) = \R_{\LL}^\text{right}(w) w.
$$
Throughout the remainder of the paper, we will only deal with \emph{left} return words.
Thus, when the language $\LL$ is clear, we will write $\R(w)$ instead of $\R_{\LL}^\text{left}(w)$.
Evidently, our results can be rephrased in terms of right return words.

\subsection{Morphisms}
\label{sec:morphisms}

A \emph{morphism} between two alphabets $\A$ and $\B$ is a map $\varphi: \A^* \to \B^*$ such that
$\varphi(\varepsilon) = \varepsilon$
and
$\varphi(uv) = \varphi(u) \varphi(v)$
for every $u,v \in \A^*$.

Given two distinct letters $a,b \in \A$, we define the morphisms
$\alpha_{a,b}, \tilde{\alpha}_{a,b} :\A^* \to \A^*$
as
$$
    \alpha_{a,b} =
    \begin{cases}
        a \mapsto a b \\
        c \mapsto c, & c \neq a
    \end{cases}
    \qquad \text{and} \qquad
    \tilde{\alpha}_{a,b} =
    \begin{cases}
        a \mapsto b a \\
        c \mapsto c, & c \neq a
    \end{cases}.
$$
With an abuse of notation, we also define the same morphisms from $\A^*$ to $\left( \A \sqcup \{b \} \right)^*$if $b \notin \A$.
Given two disjoint alphabets $\A, \B$ let
$\iota_{\A,\B}: \A^* \to \left(\A \sqcup \B \right)^*$
be defined as the identity restricted to $\A^*$.

\subsection{Permutations}
\label{sec:permutations}

A \emph{permutation} of $\A = \{ a_1 < \ldots <a_k \}$ is a bijective automorphism of $\A$.
We denote the set of such permutations by $S_\A$.

To describe an element $\pi \in S_{\A}$, we will use either the one-line notation or the cyclic notation.
For instance, the symmetric permutation defined by
$\pi(a_i) = a_{k-i+1}$ for every $1 \le i \le k$
will be denoted either as
$$
    (a_k, a_{k-1}, \ldots, a_1)
$$
or as the composition of the $2$-cycles
$$
    (a_1 a_k) (a_2 a_{k-1}) \cdots (a_{\frac{k}{2}} a_{\frac{k}{2}+1})
$$
when $k$ is even, and
$$
    (a_1 a_k) (a_2 a_{k-1}) \cdots (a_{\frac{k-1}{2}} a_{\frac{k+1}{2}+1} )
(a_{\frac{k+1}{2}})
$$
when $k$ is odd.

A permutation is \emph{circular} if it has only one cycle.
It is \emph{reducible} if the ordered alphabet $\A_{i,j} = \{ a_i < a_{i+1} < \ldots < a_j \}$ is invariant under $\pi$ for some $i \le j$ with $j-i < k-1$, i.e., $i > 0$ or $j < k$.
Note that our definition of reducibility of a permutation generalizes the classical one where $i = 0$.
Evidently, if $\A_{i,j}$ is invariant under $\pi$, then so is the alphabet $\{ a_1 < \ldots < a_{i-1} < a_{j+1} < \ldots < a_k \}$.

\subsection{Burrows-Wheeler transform}
\label{sec:bwt}

The Burrows-Wheeler transform is the map sending a word $w \in \A^*$ to the word $\bwt{\A}{w}$ obtained by concatenating the last (not necessarily distinct) letters of the $|w|$ conjugates of $w$ sorted lexicographically on $\A$.
Explicitly, let $w = w_0 \cdots w_{n-1}$ be a word of length $n$ over the alphabet $\A$, and let
$w^{(0)} \le w^{(1)} \le \ldots \le w^{(n-1)}$ be its conjugates.
Then
$\bwt{\A}{w} =
w^{(0)}_{n-1} w^{(1)}_{n-1} \cdots w^{(n-1)}_{n-1}$
is the last column of the $n \times n$ matrix
$$
    M_{\A} (w) =
    \left[
    \begin{matrix}
        w^{(0)}_{0} & w^{(0)}_{1} & \cdots & w^{(0)}_{n-1} \\
        w^{(1)}_{0} & w^{(1)}_{1} & \cdots & w^{(1)}_{n-1} \\
        \vdots & \vdots & \ddots & \vdots \\
        w^{(n-1)}_{0} & w^{(n-1)}_{1} & \cdots & w^{(n-1)}_{n-1}
    \end{matrix}
    \right].
$$

\begin{example}
\label{ex:levkovy}
Consider the word
$w = {\tt levkoy}$
over the standard ordered English alphabet
$\E = \{ {\tt a} < {\tt b} < \ldots < {\tt z} \}$.
We compute its matrix of cyclic rotations, sorted lexicographically by row, from top to bottom. 
$$
    M_{\E}(w)
    =\left[
    \begin{matrix}
        {\tt e} & {\tt v} & {\tt k} & {\tt o} & {\tt y} & {\tt l} \\
        {\tt k} & {\tt o} & {\tt y} & {\tt l} & {\tt e} & {\tt v} \\
        {\tt l} & {\tt e} & {\tt v} & {\tt k} & {\tt o} & {\tt y} \\
        {\tt o} & {\tt y} & {\tt l} & {\tt e} & {\tt v} & {\tt k} \\
        {\tt v} & {\tt k} & {\tt o} & {\tt y} & {\tt l} & {\tt e} \\
        {\tt y} & {\tt l} & {\tt e} & {\tt v} & {\tt k} & {\tt o}
    \end{matrix}
    \right].
$$

Concatenating the letters of the last column of $M_{\E}(w)$, beginning with the first row, we obtain
$\bwt{\E}{w} = {\tt lvykeo}$.
\end{example}

\begin{example}
\label{ex:peterbald}
Let us consider the two words
$w' = {\tt peterbald}$
and $w'' = {\tt bambino}$
over the same alphabet
$\E$ as in Example~\ref{ex:levkovy}.
We compute $\bwt{\E}{w'} = {\tt brltpadee}$ and
$\bwt{\E}{w''} = {\tt bombain}$.

\end{example}

The Burrows-Wheeler transform of the word $\bwt{\E}{w} = \tt lvykeo$ in Example \ref{ex:levkovy} has six distinct runs of length one. In Example \ref{ex:peterbald}, $\bwt{\E}{w} = \tt brltpadee$ has seven runs of length one and one run of length two: $\tt ee$.

The following results are well known (see, e.g.,\cite{ChrochemoreDesarmenienPerrin05,MantaciRestivoRosoneSciortino07,MantaciRestivoSciortino03}).

\begin{proposition}
\label{pro:bwt-conjugates}
Two words $u,v$ over the same ordered alphabet $\A$ are conjugate if and only if
$\bwt{\A}{u} = \bwt{\A}{v}$.
\end{proposition}

\begin{proposition}
\label{pro:bwt-primitive}
Let $u \in \A^*$ and $p \in \N$.
A word $w$ is a conjugate of $u^p$ if and only if
$\bwt{\A}{u} = b_1 \cdots b_{|u|}$
and
$\bwt{\A}{w} = b_1^p \cdots b_{|u|}^p$,
with $b_i \in \A$.
\end{proposition}

In~\cite{MantaciRestivoRosoneSciortino07} it is shown that, given an ordered alphabet $\A$, an extended version of the Burrows-Wheeler transform, denoted $\textrm{ebwt}$, gives a bijection between $\A^*$ and the multiset of Lyndon words (equivalently, to the set of powers of Lyndon words) over $\A$, where the conjugates, possibly of different lengths, are ordered using the $\omega$-order instead of the lexicographic order: $u \le_{\omega} v$ if $u^\omega \le v^\omega$.

\begin{example}
\label{ex:ebwt}
Let $W$ be the multiset $\{ {\tt aac}, {\tt ab}, {\tt ab} \}$ of Lyndon words over the alphabet $\A = \{ {\tt a} < {\tt b} < {\tt c} \}$.
The Parikh vector of $W$ is $\Psi_{\A}(W) = (4,2,1)$.
The extended Burrows-Wheeler transform of the multiset is
$\ebwt{\A}{W} = {\tt c bb aaaa}$.
Indeed, the conjugates of $W$ are:
${\tt aac} <_\omega {\tt ab} \le_\omega {\tt ab} <_\omega {\tt aca} <_\omega {\tt ba} \le_\omega {\tt ba} <_\omega {\tt caa}$.
\end{example}

Let $\pi$ be a permutation over an ordered alphabet $\A$.
A word $w \in \A^{*}$ is said to be $\pi$\emph{-clustering} for $\A$ if
$$
    \bwt{\A}{w} = a^{k_1}_{\pi(a_1)} \cdots a^{k_r}_{\pi(a_k)},
$$
where $k_i = |w|_{\pi(a_i)}$ and $k =\card{\A}$.
A word is \emph{perfectly clustering} (for its alphabet) if it is $\pi$-clustering with $\pi \in S_\A$ the symmetric permutation.
Thus, a word is clustering if its image under the Burrows-Wheeler transform is \emph{clustered}, i.e., the number of runs in it corresponds to the number of distinct letters.

Note that it is also possible to define the clustering property in terms of two permutations (starting with an unordered alphabet), as done, e.g., in~\cite{ferenczi2025clustering}.

\begin{example}
\label{ex:clustering}
Let us consider the two words $w$,$w'$ and $w''$ in Examples~\ref{ex:levkovy} and~\ref{ex:peterbald}.
The words $w$ and $w'$ are clustering, while $w''$ is not.
\end{example}

The notions of clustering and perfectly clustering can be extended to multisets of words.

\begin{example}
The multiset $W$ of Example~\ref{ex:ebwt} is (perfectly) clustering, while the multiset
$W' = \{ {\tt aab}, {\tt ab}, {\tt ab} \}$,
defined over the same ordered alphabet, is not.
Indeed, one has
$\ebwt{\A}{W'} = {\tt babbaaa}$,
since the conjugates of $W'$ are:
${\tt aab} <_\omega {\tt aba} \le_\omega {\tt ab} <_\omega {\tt ab} <_\omega {\tt baa} \le_\omega {\tt ba} <_\omega {\tt ba}$.
\end{example}

The clustering property evidently depends not only on the word (resp., multiset of words), but on the order of the alphabet as well.

\begin{example}
\label{ex:banana3}
Let us consider the three alphabets
$$
    \A = \{ {\tt a} < {\tt b} < {\tt n} \},
    \quad
    \A' = \{ {\tt a} < {\tt n} < {\tt b} \},
    \quad \text{and} \quad
    \A'' = \{ {\tt n} < {\tt a} < {\tt b} \}.
$$
The word $w = {\tt banana}$ is pangrammatic for all three alphabets.
One has
$$
    \bwt{\A}{w} = {\tt nnbaaa},
    \quad
    \bwt{\A'}{w} = {\tt bnnaa}
    \quad \text{and} \quad
    \bwt{\A''}{w} = {\tt aabna}.
$$
So, $w$ is perfectly clustering for $\A$ and $\A'$, but not clustering for $\A''$.
\end{example}

Over a binary alphabet (perfectly) clustering words coincide with powers of Christoffel words and their conjugates (see~\cite{MantaciRestivoRosoneSciortino07}).
A characterization over larger alphabets in terms of factorization into palindromes is
given in~\cite{LapointeReutenauer24} (see also~\cite{SimpsonPuglisi08}).

The following result follows immediately from Propositions~\ref{pro:bwt-conjugates} and~\ref{pro:bwt-primitive}.

\begin{proposition}
\label{pro:clusteringprimitive}
Let $w = u^p \in \A^*$ with $u$ a primitive word and $p \in \N$.
Then $w$ is $\pi$-clustering for $\A$ if and only if $u$ is $\pi$-clustering for $\A$.
\end{proposition}

\section{Interval Exchange Transformations}
\label{sec:iets}

By a (real) interval we mean a left-closed and right-open interval over the real line.
Let $\A$ be an ordered alphabet of cardinality $k$ and $\pi$ a permutation on $\A$.
An ordered partition $(I_a)_{a \in \A}$ of an interval $I$ is a finite sequence of mutually disjoint sub-intervals of $I$ such that $I_a$ is to the left of $I_b$ when $a < b$, and such that $\bigcup_{a \in \A} I_a = I$.
We assume that $|I_a|>0$ for every letter $a \in \A$.

The $k$-\emph{interval exchange transformation} (or $k$-\emph{IET} or just \emph{IET} in short) $T$ associated with a partition $(I_a)_{a \in \A}$ of $I = [\ell, r)$ and a permutation $\pi \in S_\A$ is the piecewise translation defined by
$$
    T(x) = x + \tau_a
    \quad \text{if} \quad
    x \in I_a,
$$
where
$$
    \tau_a =
    \sum_{\pi^{-1}(b) < \pi^{-1}(a)} |I_b| - \sum_{b < a} |I_b|.
$$
We define $D(T) = \displaystyle \left\{ \sum_{b<a} |I_b| \; : \; a \in \A \right\} \setminus \{ \ell \}$ as the set of \emph{formal discontinuities} of $T$.
The \emph{orbit} of a point $x \in I$ under $T$ is the set $\{T^k(x) \; | \; k \in \Z \}$.
An IET is \emph{periodic} if the orbit of any point $x \in I$ is finite.
It is \emph{minimal} if the orbit of any point $x \in I$ is dense in the interval (with respect to the standard Euclidean topology).
It is \emph{regular}, or satisfies the \emph{Keane condition} or \emph{i.d.o.c.}, if the orbits of the formal discontinuities are infinite and disjoint.
A regular IET is minimal and aperiodic (\cite{Keane75}), while the inverse is not true (see, e.g.,~\cite{bifixcodesIETs}).

A \emph{connection} of an IET $T$ is a triple $(x,y,n)$ where $x\in D(T^{-1})$, $y \in D(T)$, $n \ge 0$ and $T^n(x) = y$.
When $n=0$, we call the point $x=y$ a $0$-\emph{connection}.

Given an IET $T$, define the set
$$
    C(T) = \left\{ T^i(x) \, : \, (x,y,n) \text{ is a connection}, i \ge 0 \right\}.
$$
The set $C(T)$ is finite, since the alphabet is finite, and thus $D(T)$ is finite too.
Clearly, $C(T)$ is empty if and only if $T$ is regular.

Given an IET $T$ on $I$, we can assign an infinite word $\Omega_T(x) = w_0 w_1 w_2 \cdots$ describing the orbit of each point $x \in I$ by setting
$w_i = a$
if
$T^i(x) \in I_a$.
This word is called the \emph{trajectory} of $x$ under $T$ or the \emph{natural coding} of $T$ relative to $x$.
When $T$ is clear from the context, we write $\Omega(x)$ instead of $\Omega_T (x)$.
The \emph{language} of an IET $T$ is
$$
    \LL(T) =
    \bigcup_{x \in I} \LL(\Omega(x)).
$$
When $T$ is minimal or has only one periodic component (i.e., there is only one possible trajectory up to a shift), $\LL(T)$ does not depend on the choice of $x$.
Moreover, in this case $\LL(T)$ is uniformly recurrent (see, e.g.,~\cite{branching}) and thus recurrent.


Given an IET $T$ over $[\ell, r)$ and a word
$w = w_0 w_1 \cdots w_{n-1} \in \LL(T)$,
we define the \emph{cylinder} corresponding to $w$
$$
    I_w =
    I_{w_0} \cap T^{-1}(I_{w_1}) \cap \cdots \cap T^{-(n-1)}(I_{w_{n-1}}).
$$
By convention, we have
$I_\varepsilon = [\ell, r)$.
For every point $x \in I_w$, the trajectory $\Omega(x)$ has $w$ as a prefix.

Let $T$ be a $k$-IET on interval $[\ell, r)$.
We recall the following definitions from~\cite{branching}.
Let $J = [u,v) \subseteq [\ell,r)$.
For a point $z \in [\ell,r)$ the \emph{forward} and \emph{backward return times} of $z$ to $J$ are respectively
$$
    \rho^{+}_{J,T}(z) =
    \min\{ n > 0 \; : \; T^{n}(z) \in (u,v) \}
    \; \text{and} \;
    \rho^{-}_{J,T}(z) =
    \min\{ n \ge 0 \; : \; T^{-n}(z) \in (u,v)\}.
$$
We write
$$
    E_{J,T}(z) =
    \{ -\rho^{-}_{J,T}(z), \ldots , \rho^{+}_{J,T}(z)-1 \}
$$
and
$$
    N_{J,T}(z) =
    \{ T^{i}(z) \; : \; i \in E_{J,T}(z) \}.
$$
Finally, we set
$$
    \text{Div}(J,T) =
    \bigcup_{\gamma \in D(T)} N_{J,T} (\gamma).
$$
We say the interval $J = [u,v)$ is \emph{admissible} if $u,v \in \text{Div}(J,T) \cup \{ r \}$.
Every cylinder $I_w$ with $w \in \LL(T)$ is admissible (\cite{branching}).

A \emph{region} of an IET $T$ on the interval $I$ is a sub-interval of $I$ between two consecutive $0$-connections or between $\ell$ and the first $0$-connection or between the last $0$-connection and $r$.
Following~\cite{Boshernitzan88}, we define a \emph{component} of $T$ as the union of the sub-intervals having as endpoints consecutive points in $C(T) \cup \{0, \ell\}$ whose interior is intersected by the orbit of a point $x \in I$. Such components comprise equivalence classes of points in $I$. On each component, $T$ is either minimal or periodic.

It is known that every IET can be decomposed into minimal and periodic components (see, e.g.,~\cite{Boshernitzan88,Mayer43}).

\begin{example}
\label{ex:non-minimal-IET}
Let $T$ be the IET on the interval $[0, 1)$ associated with the partition indexed by $\{ {\tt a} < {\tt b} < {\tt c} < {\tt d} < {\tt e} \}$, with
$|I_{\tt a}| = |I_{\tt d}| = |I_{\tt e}| = 1/6$, $|I_{\tt b}| = (1-\alpha)/2$, $|I_{\tt c}| = \alpha/2$ and $\alpha < 1$ an irrational; and permutation is $\pi = ({\tt e, c, b, d, a})$ (see Figure~\ref{fig:non-minimal-IET}).

The IET $T$ has three components:
$I_{\tt a} \cup I_{\tt e}$,
$I_{\tt d}$ and
$I_{\tt b} \cup I_{\tt c}$,
the first two being periodic and the last minimal.
It has three $0$-connections: $\frac{1}{6}$, $\frac{2}{3}$ and $\frac{5}{6}$ and no other connection.
The four regions are
$I_{\tt a} = T(I_{\tt e})$,
$I_{\tt b} \cup I_{\tt c} = T(I_{\tt c}) \cup T(I_{\tt b})$,
$I_{\tt d} = T(I_{\tt d})$ and
$I_{\tt e} = T(I_{\tt a})$.

Note also that $\pi$ is reducible since $\{ {\tt b} < {\tt c} \}$ and $\{ {\tt d} \}$ are invariant under the permutation.
\end{example}

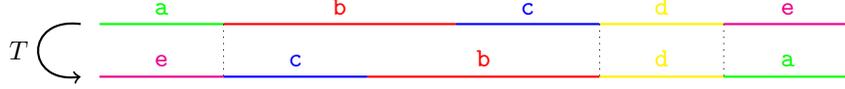
\begin{figure}[ht]
    \centering
    \begin{tikzpicture}[x=5cm, y=1cm]
        \draw[thick,green] (0,0) -- (0.33,0) node[midway,above] {\tt a};
        \draw[thick,red] (0.33,0) -- (0.33+1-0.382,0) node[midway,above] {\tt b};
        \draw[thick,blue] (0.33+1-0.382,0) -- (0.33+1,0) node[midway,above] {\tt c};
        \draw[thick,yellow] (0.33+1,0) -- (2*0.33+1,0) node[midway,above] {\tt d};
        \draw[thick,magenta] (2*0.33+1,0) -- (2,0) node[midway,above] {\tt e};
        
        \draw[thick,magenta] (0,-0.7) -- (0.33,-0.7) node[midway,above] {\tt e};
        \draw[thick,blue] (0.33,-0.7) -- (0.33+0.382,-0.7) node[midway, above] {\tt c};
        \draw[thick,red] (0.33+0.382,-0.7) -- (0.33+1,-0.7) node[midway,above] {\tt b};
        \draw[thick,yellow] (0.33+1,-0.7) -- (2*0.33+1,-0.7) node[midway,above] {\tt d};
        \draw[thick,green] (2*0.33+1,-0.7) -- (2,-0.7) node[midway,above] {\tt a};

        \draw[dotted] (0.33,0) -- (0.33,-0.7) {};
        \draw[dotted] (1.33,0) -- (1.33,-0.7) {};
        \draw[dotted] (1.66,0) -- (1.66,-0.7) {};
        
        \draw[->, thick] (-.05, 0) .. controls (-0.2, 0.1) and (-0.2, -0.8) .. (-.05, -0.7) node[midway, left] {$T$};
    \end{tikzpicture}
    \caption{An IET with three components and four regions (for simplicity we denote with the letter $a$ both by $I_a$ and $T(I_a)$.}
    \label{fig:non-minimal-IET}
\end{figure}

\begin{example}
\label{ex:coding-non-minimal-IET}
Let $T$ be the IET defined in Example~\ref{ex:non-minimal-IET}, with $\alpha = \frac{3-\sqrt{5}}{2}$.
One can check that
$$
    \LL(T) =
    \LL(({\tt ae})^\omega) \cup \LL(({\tt d})^\omega) \cup \LL(\ff),
$$
where
$\ff = {\tt bcbbcb}\cdots$
is the well-known Fibonacci word over the alphabet
$\{ {\tt b} < {\tt c} \}$.
\end{example}

The \emph{discrete interval exchange} (or \emph{DIET} in short) 
associated with the composition
$(n_1, n_2, \ldots, n_k)$
of
$n = \sum_{i=1}^k n_i$
and the permutation
$\pi \in S_k$
is the map
$$
    T(h) = h + t_i,
    \quad \text{if} \quad
    \sum_{j < i} n_j < h \le \sum_{j \le i} n_j,
$$
where
$$
    t_i =
    \sum_{\pi^{-1}(j) < \pi^{-1}(i)} n_j \; - \; \sum_{j < i} n_j.
$$
Every such a DIET
correspond to a standard IET associated with a partition
$(I_{a})_{a \in \A}$
and a permutation
$\pi \in S_\A$,
where
$\A = \{ a_1 < \ldots < a_d \}$
and
$|I_{a_i}| = n_i$
(see Figure~\ref{fig:7DIET} below).
Note that each component of this corresponding IET is periodic;
thus, in particular, this associated IET cannot be minimal or regular.

Cylinders for DIETs are defined analogously to the IET case.

There is a strong connection between clustering multisets of words and DIETs.
In fact, if a multiset $W \subset \A^*$ is $\pi$-clustering, then its Parikh vector yields a composition of
$n = \displaystyle \sum_{w \in W} |w|$
that, along with $\pi$, defines a DIET.
Similarly to IETs, we can encode the (periodic) trajectories by encoding each integer
$\displaystyle h \in \left[ \sum_{j<i} n_j, \, \sum_{j \le i} n_j \right]$ by the $i^{\text{th}}$ letter of the alphabet.

In a symmetric way, it is possible to show that every DIET corresponds to a unique multiset of Lyndon words.

\begin{example}
\label{ex:ebwt-DIET}
Let $W$ be the multiset of Example~\ref{ex:ebwt}.
We can define a DIET $T$ associated with the composition
$(4,2,1)$ of $7$
and the permutation
$\pi = ({\tt c}, {\tt b}, {\tt a})$.
The action of the DIET over
$\{ {\tt 1}, {\tt 2}, \ldots, {\tt 7} \}$
is given by
$$
    \mu = ({\tt 1,4,7})({\tt 2,5})({\tt 3,6}) \in S_7,
$$
where one can clearly identify $S_7$ with $S_{\B}$ with $\B = \{ {\tt 1} < {\tt 2} < \cdots < {\tt 7} \}$ (see left of Figure~\ref{fig:7DIET}).
Note that each cycle corresponds to one of the primitive words in $W$.
For instance, the trajectory of ${\tt 4}$ is given by
$\Omega({\tt 4}) = ({\tt aca)}^\omega$,
i.e., the infinite repetition of a conjugate of ${\tt aac}$.
One can check that
$$
    I_{\tt a} = \{ {\tt 1,2,3,4} \},
    \quad
    I_{\tt ab} = \{ {\tt 2,3} \},
    \quad \text{and} \quad
    I_{\tt aac} = \{ {\tt 1} \}.
$$
The corresponding IET is shown on the right of Figure~\ref{fig:7DIET}.
\end{example}

\begin{figure}[ht]
    \centering
    \begin{tikzpicture}[scale=0.5]
        \node (u1) {\tt 1};
        \node (u2) [right=0cm of u1] {\tt 2};
        \node (u3) [right=0cm of u2] {\tt 3};
        \node (u4) [right=0cm of u3] {\tt 4};
        \node (u5) [right=0cm of u4] {\tt 5};
        \node (u6) [right=0cm of u5] {\tt 6};
        \node (u7) [right=0cm of u6] {\tt 7};
        \node[draw,rounded corners,blue,minimum width=1.6cm, minimum height=0.4cm, right=-0.35cm of u1] (ua) {};
        \node[draw,rounded corners,red,minimum width=0.75cm, minimum height=0.4cm, right=0.05cm of ua] (ub) {};
        \node[draw,rounded corners,green,minimum width=0.3cm, minimum height=0.4cm, right=0.15cm of ub] (uc) {};

        \node (b1) [below=0.1cm of u1] {\tt 1};
        \node (b2) [right=0cm of b1] {\tt 2};
        \node (b3) [right=0cm of b2] {\tt 3};
        \node (b4) [right=0cm of b3] {\tt 4};
        \node (b5) [right=0cm of b4] {\tt 5};
        \node (b6) [right=0cm of b5] {\tt 6};
        \node (b7) [right=0cm of b6] {\tt 7};
        \node[draw,rounded corners,green,minimum width=0.3cm, minimum height=0.4cm, right=-0.35cm of b1] (bc) {};
        \node[draw,rounded corners,red,minimum width=0.7cm, minimum height=0.4cm, right=0.1cm of bc] (bb) {};
        \node[draw,rounded corners,blue,minimum width=1.6cm, minimum height=0.4cm, right=0.15cm of bb] (ba) {};
        
        \draw[->, thick] (-.4, 0) .. controls (-.9, -0.1) and (-.9, -0.9) .. (-.4, -1) node[midway, left] {$T$};
    \end{tikzpicture}
    \qquad
    \begin{tikzpicture}[scale=0.5]
        \draw[thick,blue] (0,0.5) -- (4,0.5) node[midway,above] {\tt a};
        \draw[thick,red] (4,0.5) -- (6,0.5) node[midway,above] {\tt b};
        \draw[thick,green] (6,0.5) -- (7,0.5) node[midway,above] {\tt c};

        \draw[thick,green] (0,-0.7) -- (1,-0.7) node[midway,above] {\tt c};
        \draw[thick,red] (1,-0.7) -- (3,-0.7) node[midway, above] {\tt b};
        \draw[thick,blue] (3,-0.7) -- (7,-0.7) node[midway,above] {\tt a};

        \draw[->, thick] (-.3, 0.5) .. controls (-.9, 0.4) and (-.9, -0.6) .. (-.3, -0.7) node[midway, left] {$T$};
    \end{tikzpicture}
    \caption{A DIET (on the left) and its associated IET (on the right).}
    \label{fig:7DIET}
\end{figure}
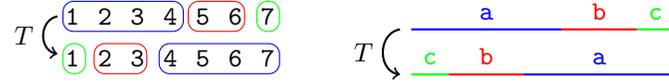

In particular, one can interpret every primitive $\pi$-clustering word $w \in \A^*$ as a DIET associated with the composition $\Psi_{\A}(w)$ of $|w|$ and the permutation $\pi$, the permutation $\mu$ describing the action of such a DIET being circular (see~\cite{FerencziZamboni13} for a characterization of $\pi$-clustering words in terms of trajectories in IETs or DIETs).




\section{Induced Interval Exchange Transformations}
\label{sec:Rauzy}

Let $\A = \{ a_1 < \ldots < a_k \}$ be an ordered alphabet and let $\pi \in S_{\A}$ be a permutation on $\A$.
Let $T$ be an IET over $I = [\ell, r)$ associated with a partition $(I_a)_{a \in \A}$ of $I$ and a permutation $\pi \in S_\A$.
The \emph{transformation induced} by $T$ on a sub-interval $J \subset I$ is the map $T': J \to J$ defined by $T'(z) = T^{\nu(z)} (z)$, where
$$
    \nu(z) = \min \{ n>0 \; : \; T^n(z) \in J \}
$$
is called the \emph{first return map} of $T$ to $J$.
Note that $\nu(z)$ is well-defined because IETs do not have wandering intervals (see, e.g.,~\cite{FerencziHubertZamboni24}).

We now develop the induction machinery necessary to extend the results of~\cite{Rauzy79} and~\cite{branching}.
From any standard IET $T$, we obtain "related" IETs, i.e., IETs associated either with the same alphabet $\A$ (possibly reordered) or with a sub-alphabet of it (possibly circularly reordered). We do this via an analysis of IET first-return map dynamics.

Assume first that $\pi$ is irreducible.

\subsection{Right Rauzy step}
\label{sec:right}

A \emph{right Rauzy step} is a mapping $\rho$ sending $T$ to the induced transformation $T'$ on $I' = [\ell, r')$, where $r'$ is the rightmost among the points in $D(T) \cup D(T^{-1})$.

Assume that $\pi(a_k) \neq a_k$.

If $r'$ is not a connection, that is, if $|I_{a_k}| \ne |I_{\pi(a_k)}|$, we call the step \emph{non-degenerate}.
This case corresponds to the classical right Rauzy induction introduced by Rauzy in~\cite{Rauzy79}.

If $|I_{a_k}| > |I_{\pi(a_k)}|$ then $\Omega(x)$ starts with $\pi(a_k) a_k$ for every $x \in I_{\pi(a_k)}$.

If $|I_{a_k}| < |I_{\pi(a_k)}|$ then $\Omega(x)$ starts with $\pi(a_k) a_k$ for every $x \in T^{-1}(I_{a_k})$.

We first consider the former case (i.e., when $|I_{a_k}| > |I_{\pi(a_k)}|$).
Let $h = \pi^{-1} (k)$.
Then $\rho (T)$ is the IET associated with $\mathcal{A}$ and $\pi' \in S_{\mathcal{A}}$ defined as
\begin{equation}
    \label{eq:pik}
    \pi'(i) =
    \begin{cases}
        \pi(i) & \text{if } i \le h \\
        \pi(k) & \text{if } i = h + 1 \\
        \pi(i) + 1 & \text{if } i > h + 1
    \end{cases}
    .
\end{equation}
In the latter case
(i.e., when $|I_{a_k}| < |I_{\pi(a_k)}|$)
$\rho (T)$ is the IET associated with
$\A' = \{a_1 < \ldots < a_h < a_k < a_{h+1} < \ldots < a_{k-1} \}$
and
$\pi \in S_{\A'}$
defined as
\begin{equation}
    \label{eq:pipik}
    \pi' (i) = \pi (i)
    \quad \text{for all }
    1 \le i \le k.
\end{equation}
It is well-known that if $T$ is regular, we can apply this step infinitely many times and always obtain a new regular IET over the same alphabet (possibly reordered).

Let us now consider the \emph{degenerate} case, which occurs when $r'$ is a connection.
Equivalently, the degenerate case occurs when we have $|I_{a_k}| = |I_{\pi(a_k)}|$.
In such a case, $\pi$ can be either reducible or irreducible.


For every point $x \in I_{\pi(a_k)}$, the trajectory $\Omega(x)$ starts with $\pi(a_k) a_k$.
The first return map of $T$ to $I' = [\ell, r')$ coincides with $T$ for every point $x \notin I_{\pi(a_k)}$ and with $T^2$ for $x \in I_{\pi(a_k)}$.
In this case, we define $\rho(T)$ to be the IET on $I'$, partitioned by $(I_a)_{a \in \A'}$, with the permutation $\pi' \in S_{\A'}$, where
$\A' = \{ a_1 < \ldots < a_{k-1} \}$
and
\begin{equation}
    \label{eq:pi'deg}
    \pi'(a) =
    \begin{cases}
        \pi(a_k) & \text{if } a = \pi^{-1}(a_k) \\
      \pi(a) & \text{otherwise}
    \end{cases}.
\end{equation}

The following proposition shows that this new map is truly an IET and behaves as expected, including in the degenerate case.

\begin{proposition}
\label{pro:well-behaved-merging}
Let $T$ be a $k$-IET on the interval $I$ over the alphabet $\A$ with associated permutation $\pi \in S_\A$.
Suppose that $|I_{a_k}| = |I_{\pi(a_k)}|$ and $a_k \neq \pi(a_k)$.
Then, the map $\rho(T)$ on $I' = [\ell, r - |I_{a_k}|)$ over the alphabet $\A' = \A \setminus \{a_k\}$ with the permutation $\pi'$ as in Equation~\ref{eq:pi'deg} is a $(k-1)$-IET.

Moreover, for all $x \in I'$,
$\Omega_T(x) = \varphi(\Omega_{\rho(T)}(x))$,
where $\varphi$ is the morphism
$\varphi: \A'^\omega \rightarrow \A^\omega$
defined by $\varphi(\pi(a_k)) = \pi(a_k) a_k$ and $\varphi(c) = c$ for every letter $c \ne \pi(a_k)$.
\end{proposition}
\begin{proof}
Consider the first-return map $\nu$ of $T$ to $I$.

If $x \in I_{\pi (a_k)} \subset I'$, then $T(x) \in I_{a_k} \not\subset I'$, but $T^2 (x) \in I'$.
Thus, $\nu (x) = 2$.

Otherwise, if $x \in I_{a}$ with $a \neq \pi (a_k)$, then $T(x) \in T(I_a) \subset I'$.
Thus, $\nu(x) = 1$.

Let us denote this first-return map by $T'$.
For each each $x \in I_a \subset I'$, with $a \neq \pi (a_k)$, we have $T'(x) = T(x)$.
For all $x \in I_{\pi (a_k)}$, we have $T'(x) = T^2(x)$.
Hence,
$$
    T'(x) =
    \begin{cases}
        T^2(x) & \text{if } x \in I_{\pi(a_k)} \\
        T(x) & \text{otherwise}
    \end{cases}.
$$
Consequently, we obtain the induced map on $I'$ with the domain alphabet $\A' = \A \setminus \{a_k\}$.
In particular, $T'$ is an IET over the alphabet $ \A' $, with $|\A'| = k-1$.
Thus, $T'$ is a $(k-1)$-IET.

Reading the image sub-intervals of $T'$ in $I'$ gives a permutation of $\A'$ obtained precisely by removing $a_k$ from $\pi$ and replacing the single occurrence of $a_k$ in the image of $\pi$ by $\pi(a_k)$, and removing the $\pi(a_k)$ in the image of $\pi$.
This is exactly the definition of $\pi'$.

Thus, $T' = \rho (T)$ is a $(k-1)$-IET on $I'$ over the alphabet $\A'$ with permutation $\pi'$, proving the first part of the statement.

Let us now consider a point $x \in I'$.
The trajectory of $x$ is
$\Omega_{\rho(T)}(x) = c_0 c_1 c_2 \cdots$, with every $c_n \in \A'$.
Let us define the sequence of return times $\nu_0 = 0$ and
$$
    \nu_{n+1} =
    \begin{cases}
        \nu_n + 2 & \text{if } c_n = \pi (a_k) \\
        \nu_n + 1 & \text{if } c_n \neq \pi (a_k)
    \end{cases}.
$$
By construction, $T^{\nu_n} (x) = T'^n (x) \in I'$, and there is no visit to $I'$ at any intermediate time between $\nu_n$ and $\nu_{n+1}$.
We claim that
$$
    \Omega_T (x) =
    \varphi(c_0) \varphi(c_1) \varphi(c_2) \cdots =
    \varphi(\Omega_{\rho(T)}(x)).
$$
In particular, we claim that for each $n \ge 0$, the factor of $\Omega_T(x)$ from position $\nu_n$ up to position $\nu_{n+1} - 1$ equals $\varphi(c_n)$.
To show the equality, we deal with the following two cases.

First, suppose $c_n \neq \pi(a_k)$.
An application of $T$ to $T'^n(x) \in I_{c_n}$ results in a $T(x)$ which remains in $I'$.
Consequently, the return to $I'$ occurs at time $\nu_n + 1$ and the letter at position $\nu_n$ is exactly $\varphi(c_n)$.

Now suppose $c_n = \pi (a_k)$.
We have $T(I_{\pi (a_k)}) = I_{a_k} \not\subset I'$.
The first application of $T$ to an $x \in I_{\pi (a_k)}$ therefore results in an occurrence of $\pi (a_k)$ in the trajectory.
The second application of $T$ returns $x$ to $I'$ and results in an occurrence of $a_k$ in the trajectory.
Hence, the two letters at positions $\nu_n$ and $\nu_{n} + 1$ are $\pi (a_k) a_k = \varphi(\pi(a_k))$, and the return occurs exactly at time $\nu_n + 2$.

Concatenating over all $n$ yields the desired equality.
\end{proof}

\begin{example}
Let $T$ be the IET defined in Example~\ref{ex:non-minimal-IET}.
The IET $\rho(T)$
is associated with $(I_a)_{a \in \A'}$ and $\pi'$, where
$\A' = \{ {\tt a} < {\tt b} < {\tt c} < {\tt d} \}$
and
$\pi' = ({\tt a c b d})$
(see Figure~\ref{fig:merging}).
\label{ex:merging}

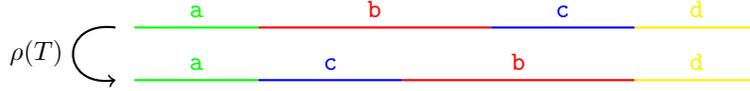
\begin{figure}[ht]
    \centering
    \begin{tikzpicture}[x=5cm, y=1cm]
       \draw[thick,green] (0,0) -- (0.33,0) node[midway,above] {\tt a};
       \draw[thick,red] (0.33,0) -- (0.33+1-0.382,0) node[midway,above] {\tt b};
        \draw[thick,blue] (0.33+1-0.382,0) -- (0.33+1,0) node[midway,above] {\tt c};
        \draw[thick,yellow] (0.33+1,0) -- (2*0.33+1,0) node[midway,above] {\tt d};

        \draw[thick,green] (0,-0.7) -- (0.33,-0.7) node[midway,above] {\tt a};
        \draw[thick,blue] (0.33,-0.7) -- (0.33+0.382,-0.7) node[midway, above] {\tt c};
        \draw[thick,red] (0.33+0.382,-0.7) -- (0.33+1,-0.7) node[midway,above] {\tt b};
        \draw[thick,yellow] (0.33+1,-0.7) -- (2*0.33+1,-0.7) node[midway,above] {\tt d};

        \draw[->, thick] (-.05, 0) .. controls (-0.2, 0.1) and (-0.2, -0.8) .. (-.05, -0.7) node[midway, left] {$\rho(T)$};
    \end{tikzpicture}
    \caption{The IET in Example~\protect{\ref{ex:non-minimal-IET}}.}
    \label{fig:merging}
\end{figure}
\end{example}

\begin{remark}
Another possible choice would be to "preserve" $a_k$ instead of $\pi(a_k)$.
The procedure would be similar, with a simpler description of $\pi'$, but a more complicated description of $\A'$.
\end{remark}

\subsection{Left Rauzy step}
\label{sec:left}

We symmetrically define the \emph{left Rauzy step} $\lambda$ sending $T$ to the induced transformation $T'$ on $I' = [\ell', r)$, where $\ell'$ is the left-most point of $(\ell,r)$ belonging to $D(T) \cup D(T^{-1})$.

Assume $\pi(a_1) \neq a_1$.

If $\ell'$ is not a connection, that is, if $|I_{a_1}| \neq |I_{\pi(a_1)}|$, we call the step \emph{non-degenerate}.
This is the mirror of the classical right Rauzy induction of~\cite{Rauzy79} (see also~\cite{branching}).
Symmetrically as seen above, if $T$ is regular, we can apply this step infinitely many times and always obtain a new regular IET on the same (possibly reordered) alphabet.

If $|I_{a_1}| > |I_{\pi(a_1)}|$ then $\Omega(x)$ starts with $\pi(a_1)\, a_1$ for every $x \in I_{\pi(a_1)}$. 

Similarly, if $|I_{a_1}| < |I_{\pi(a_1)}|$, then $\Omega(x)$ starts with $\pi(a_1) a_1$ for every $x \in T^{-1}(I_{a_1})$. 

We first consider the former case
(i.e., when $|I_{a_1}| > |I_{\pi(a_1)}|$).
Let $h = \pi^{-1} (1)$.
Then $\lambda (T)$ is the regular IET associated with $\A$ and $\pi' \in S_\A$ defined as
\begin{equation}
    \label{eq:pikL}
    \pi'(i) =
    \begin{cases}
        \pi(i) - 1& \text{if } i < h - 1 \\
        \pi(1) & \text{if } i = h - 1 \\
        \pi(i) & \text{if } i \le h
    \end{cases}
    .
\end{equation}
In the latter case
(i.e., when $|I_{a_1}| < |I_{a_{\pi (1)}}|$)
$\lambda (T)$ is the IET associated with
$\A' = \{a_2 < \ldots < a_{h-1} < a_1 < a_h < \ldots < a_k \}$
and
$\pi \in S_{\A'}$,
where
\begin{equation}
    \label{eq:pipikL}
    \pi(i) = \pi'(i)
\end{equation}
for every $1 \le i \le k$.

Let us now consider the \emph{degenerate case}, which occurs when $\ell'$ is a connection.
Equivalently, the degenerate case occurs when we have
$|I_{a_1}| = |I_{\pi(a_1)}|$.
In such a case, $\pi$ can be either reducible or irreducible.

For every point $x \in I_{\pi(a_1)}$, the trajectory $\Omega(x)$ starts with $\pi(a_1) a_1$.
The first return map of $T$ to $I' = [\ell', r)$ coincides with $T$ for all points $x \notin I_{\pi(a_1)}$ and with $T^2$ for $x \in I_{\pi(a_1)}$.
Analogously to the right case, we define $\lambda(T)$ as the IET associated with the partition $(I_a)_{a \in \A'}$ of $I'$ and the permutation $\pi'$, with
$\A' = \{ a_2 < \ldots < a_k \}$
and
\begin{equation}
    \label{eq:pi'degL}
    \pi'(a) =
    \begin{cases}
        \pi(a_1) & \text{if } a = \pi^{-1}(a_1) \\
        \pi(a) & \text{otherwise}
    \end{cases}
    .
\end{equation}
We now provide the mirror proposition of Proposition~\ref{pro:well-behaved-merging}.
The proof is symmetric to that of the previous proposition.

\begin{proposition}
\label{pro:well-behaved-merging-left}
Let $T$ be a $k$-IET on the interval $I$ over the alphabet $\A$ with an associated permutation $\pi \in S_\A$.
Suppose that $|I_{a_1}| = |I_{\pi(a_1)}|$ and $a_1 \neq \pi(a_1)$.
Then, the map $\lambda(T)$ on $I' = [\ell + |I_{a_1}|, r)$ over the alphabet $\A' = \A \setminus \{ a_1 \}$ with the permutation $\pi'$ as in Equation~\eqref{eq:pi'degL} is a $(k-1)$-IET.

Moreover, for all $x \in I'$,
$\Omega_T(x) = \varphi(\Omega_{\lambda(T)}(x))$,
where $\varphi$ is the morphism
$\varphi: \A'^\omega \rightarrow \A^\omega$
defined by
$\varphi(\pi(a_1)) = \pi(a_1) a_1$
and $\varphi(c) = c$ for every letter $c \ne \pi(a_1)$.
\end{proposition}

\begin{remark}
\label{rmk:moving}
Right and left Rauzy steps "cut" the IET only around $\ell$, if we identify $\ell$ with $r$.
Mutatis mutandis, it is also possible to define Rauzy steps around the other $0$-connections.
Indeed, considering a $0$-connection of the form $\sum_{b < a_i} |I_{a_i}|$, it is enough to permute the alphabet circularly sending $a_j \mapsto a_{(j-i \pmod k) + 1}$.
\end{remark}

\begin{example}
\label{ex:IET}
Consider the $5$-IET $T$ depicted in Figure~\ref{fig:non-minimal-IET}.
We can circularly reorder it about the connection between $I_a$ and $I_b$, obtaining an IET $T'$ such that the leftmost interval in the domain order is now $I_b$ (see Figure~\ref{fig:non-minimal-circreordered}).

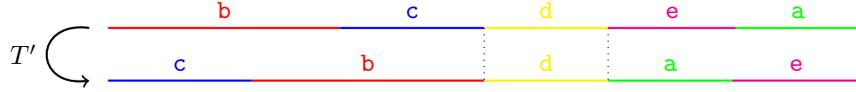
\begin{figure}[ht]
    \centering
    \begin{tikzpicture}[x=5cm, y=1cm]
    \draw[thick,red] (0,0) -- (0.618,0) node[midway,above] {\tt b};
    \draw[thick,blue] (0.618,0) -- (1,0) node[midway,above] {\tt c};
    \draw[thick,yellow] (1,0) -- (1.33,0) node[midway,above] {\tt d};
    \draw[thick,magenta] (1.33,0) -- (1.67,0) node[midway,above] {\tt e};
    \draw[thick,green] (1.67,0) -- (2,0) node[midway,above] {\tt a};

    \draw[thick,blue] (0,-0.7) -- (0.382,-0.7) node[midway, above] {\tt c};
    \draw[thick,red] (0.382,-0.7) -- (1,-0.7) node[midway,above] {\tt b};
    \draw[thick,yellow] (1,-0.7) -- (1.33,-0.7) node[midway,above] {\tt d};
    \draw[thick,green] (1.33,-0.7) -- (1.66,-0.7) node[midway,above] {\tt a};
    \draw[thick,magenta] (1.66,-0.7) -- (2,-0.7) node[midway,above] {\tt e};

    \draw[dotted] (1,0) -- (1,-0.7) {};
    \draw[dotted] (1.33,0) -- (1.33,-0.7) {};

    \draw[->, thick] (-.05, 0) .. controls (-0.2, 0.1) and (-0.2, -0.8) .. (-.05, -0.7) node[midway, left] {$T'$};
    \end{tikzpicture}
    \caption{A circularly reordered IET of Figure~\protect{\ref{fig:non-minimal-IET}}.}
    \label{fig:non-minimal-circreordered}
\end{figure}
\end{example}

\subsection{Merging}
\label{sec:merging}

To distinguish the behavior of the morphisms $\rho$ and $\lambda$ in the non-degenerate case from their behavior in the degenerate case, we define the notion of \emph{merging}.
Let $T$ be a $k$-IET on $[\ell, r)$ over the alphabet
$\A = \{a_1 < \ldots < a_k \}$
with associated permutation
$\pi \in S_\A$.

Suppose first that a right cut on $T$ would result in a degenerate cut.
In this case, we call $\rho$ a \emph{right merging step}.

Similarly, let $T$ be as before, but suppose that a left cut would result in a degenerate cut.
In this case, we call $\lambda$ a \emph{left merging step}. 

When the context makes the side on which the merging occurs clear, we drop the specification of orientation and simply call the step a \emph{merging step}.

\subsection{Splitting}
\label{sec:splitting}

We develop the \emph{splitting} step to isolate periodic components of IETs when they appear during the induction process.

Let us consider a nonempty proper contiguous sub-alphabet
$\B = \{ a_i < a_{i+1} < \ldots < a_{j} \}$
of $\A$, with $i < j$, and such that
$\bigcup_{a \in \B} I_a$
is an interval in the domain of $T$, i.e., the $I_a$ with $a \in \B$ are contiguous.
Set
$$
    I_\B =
    \bigcup_{a_i \in \B} I_{a_i} =
    \left[ \gamma_\B, \delta_\B \right),
$$
where $\gamma_\B$ and $\delta_\B$ are $0$-connections.
We say that $\B$ is \emph{$T$-invariant} if $T(I_\B) = I_\B$.

Let us also denote $\overline{\B} = \A \setminus \B$.
Supposing $T$ has such an invariant sub-alphabet $\B \subset \A$, let us consider
$$
    I_{\overline{\B}} = I \setminus I_\B =
    \begin{cases}
        [\ell, \delta_{\B} ) & \text{if } I_\B \text{ is of the form } [\delta_\B, r) \\
        [\gamma_\B, r) & \text{if } I_\B \text{ is of the form } [\ell, \gamma_\B) \\
        [\ell, \gamma_\B) \cup [\delta_\B, r) & \text{if } I_\B \text{ is interior to } I
    \end{cases}.
$$
Note that $I_{\overline{\B}}$ is not necessarily an interval.
To "transform" it into an interval in the case when $I_\B$ is interior, we connect can translate $[\delta_\B, r)$ to the left by $|I_\B|$, yielding the interval $[\ell, r - |I_\B|)$.
Formally, let us consider the \emph{gluing map} $G$
$$
    G:
    [\ell, \gamma_\B) \cup [\delta_\B, r)
    \rightarrow
    [\ell, r - |I_\B|)
$$
defined as
$$
    G(x) =
    \begin{cases}
        x & \text{if } x \in [\ell, \gamma_\B) \\
        x - |I_\B| & \text{if } x \in [\delta_\B, r)
    \end{cases}
    .
$$

When $\B \subset \A$ is $T$-invariant, we consider the three possible decompositions of $T$ into two IETs — one on $I_\B$ and one on $I_{\overline{B}}$ — corresponding to the three cases of $I_{\overline{B}}$.

Define
$$
    S_\B :=
    T\vert_{I_\B} :
    I_\B \rightarrow I_\B
$$
to be the restriction of $T$ to $I_\B$.

If $I_{\overline{B}}$ is connected, define the second IET to be the restriction of $T$ to $I_{\overline{B}}$, i.e.,
$$
    S_{\overline{\B}} :=
    T\vert_{I_{\overline{\B}}} :
    I_{\overline{\B}} \rightarrow I_{\overline{\B}}.
$$
If $I_{\overline{\B}}$ is disconnected, we define
$$
    S_{\overline{\B}} :=
    G \circ T\vert_{I_{\overline{\B}}} \circ G^{-1} :
    [\ell, r - |I_\B|) \to [\ell, r - |I_\B|).
$$

We show that the two IETs $S_\B$ and $S_{\overline{\B}}$ are well defined in all cases.

\begin{lemma}
\label{lem:wellDefinedSplitting}
Let $T: [\ell, r) \rightarrow [\ell, r)$ be a $k$-IET over the ordered alphabet $\A$.
Let $\B \subset \A$ be a proper $T$-invariant sub-alphabet whose sub-intervals $I_b$, $b \in \B$, are contiguous in $[\ell, r)$.
Then, $S_\B$ and $S_{\overline{\B}}$ as defined above are well defined IETs.
\end{lemma}
\begin{proof}
Let $b \in \B$.
Then $T$ acts on $I_b$ as an isometric translation back into (and never outside of) $I_\B$.
Furthermore, $T$ is a bijection from $I_\B$ to $I_\B$.
Thus, $T \vert_{I_\B}$ (i.e., $S_\B$ by definition) is a $|\B|$-IET. 

The same argument holds on $I_{\overline{\B}}$.
In the interior case, i.e., when $I_{\overline{\B}}$ is not an interval, we conjugate by the order-preserving translation $G$, and so $S_{\overline{\B}}$ is also an IET.
In particular, it is a $|\overline{\B}|$-IET, where $|\overline{\B}| = k - |\B|$.
\end{proof}

Furthermore, for given IET $T$ over $\A$ and proper contiguous sub-alphabet $\B \subset \A$, the $|\B|$-IET $S_\B$ and the $|\overline{\B}|$-IET $S_{\overline{\B}}$ are unique.

\begin{lemma}
\label{lem:uniquelyDeterminedSplitting}
The IETs $S_\B$ and $S_{\overline{\B}}$ are uniquely determined by $T$ and $\B$.
\end{lemma}
\begin{proof}
The domains of $S_\B$
(i.e., $I_\B$)
and $S_{\overline{\B}}$
(i.e., $I_{\overline{\B}}$ or $[\ell, r - |I_\B|)$)
are forced by the choice of $\B$ by definition.
Since $T(I_\B) = I_\B$, the restriction
$T \vert_{I_\B} = S_\B$ is unique.
This forces the restriction of $S_{\overline{\B}}$ to the complement of $I_\B$, i.e., $I_{\overline{\B}}$ in the non-internal case.
In the internal case, $G$ is an order-preserving translation, and so $S_{\overline{\B}}$ is uniquely $T \vert_{I_{\overline{\B}}}$ or its conjugate by $G$, $G \circ T\vert_{\overline{\B}} \circ G^{-1}$.
\end{proof}

Finally, we show that the permutations of $S_\B$ and $T_2$ are exactly the restrictions of the permutation of $T$ to their respective sub-alphabets.

\begin{lemma}
\label{lem:permutationPreservingSplitting}
The IETs $S_\B$ and $S_{\overline{\B}}$ have respectively the permutation as $T$ restricted to the alphabets $\B$ and $\overline{\B}$.
\end{lemma}
\begin{proof}
The top partition of $S_\B$ is $(I_b)_{b \in \B}$ in the previous left-to-right order.
Since $T (I_\B) = I_\B$, the set $I \setminus I_\B$ also disappears from the bottom ordering, which leaves exactly the restriction
$$
    \pi \vert_{\B} \in S_{\B}.
$$
Therefore, $S_\B$ has the permutation $\pi \vert_{\B}$ on $\B$.
Similarly, in the case of $S_{\overline{\B}}$, the top partition is $(I_b)_{b \in \overline{\B}}$ in the previous left-to-right order.
Since $\pi (\B) = \B$, the set $I_\B$ also disappears from the bottom ordering, which leaves exactly the restriction 
$$
    \pi \vert_{{\overline{\B}}} \in S_{\overline{\B}}.
$$
Therefore, $S_{\overline{\B}}$ has the permutation $\pi \vert_{{\overline{\B}}}$ on $\overline{\B}$.

Consequently, the permutation of $T$ is preserved on the restricted alphabets $\B$ and $\overline{\B}$ by $S_\B$ and $S_{\overline{\B}}$, respectively.
\end{proof}

The decomposition of $T$ into $S_\B$ and $S_{\overline{\B}}$ isolates invariant sub-alphabets of $T$ from the rest of $T$.
Such invariant sub-alphabets, however, are handled neither by left/right Rauzy steps nor the merging steps.

Given $T$, $\A$ and $\B$ as above, we define the \emph{splitting} step $\sigma_\B$.
This map is defined on the set of IETs having $\B$ as a proper contiguous invariant sub-alphabet, and it takes values in the set of pairs of IETs, mapping  $T \mapsto (S_\B, S_{\overline{\B}})$.
When $\B$ is clear from the context, we write $\sigma$ instead of $\sigma_\B$.

In the following proposition, we formally define this step and show that the dynamics of $\sigma(T)$ behave as expected.

\begin{proposition}
Let $T: [\ell, r) \rightarrow [\ell, r)$ be a $k$-IET over an ordered alphabet $\A$ with corresponding reducible permutation $\pi \in S_\A$.
Let $\B \subset \A$ be a proper contiguous $T$-invariant sub-alphabet of $\A$.
The multi-valued map $\sigma_\B$ is well defined, uniquely determined, and permutation preserving.
\end{proposition}
\begin{proof}
The results easily follows from the the definition of splitting and Lemmata~\ref{lem:wellDefinedSplitting},~\ref{lem:uniquelyDeterminedSplitting} and~~\ref{lem:permutationPreservingSplitting}.
\end{proof}

Associated to each $\sigma_\B$ step are two morphisms:
$$
    \iota_{\B, \A} : \B^* \rightarrow \A^*
    \quad \text{and} \quad
    \iota_{\overline{\B}, \A} :\overline{\B}^* \rightarrow \A^*
$$
that map from the free monoid over the underlying alphabets of $S_\B$ and $S_{\overline{\B}}$ to the free monoid over the underlying alphabet of the IET $T$ from which they were obtained.

\begin{example}
\label{ex:splitting}
Let $T'$ be the IET defined in Example~\ref{ex:merging} and $\B = \{ {\tt a} < {\tt b} < {\tt c} \} \subset \A$.
Then, $\sigma_\B(T') = (S_\B, S_{\overline{\B}})$, where
$S_\B$ is the IET associated with $\B$ and 
$S_{\overline{\B}}$ is the IET associated with $\overline{\B} = \{ {\tt d} \}$
(see Figure~\ref{fig:splitting-merged}).

\begin{figure}[H]
    \centering
    \begin{tikzpicture}[x=5cm,y=1cm]
        \draw[thick,green] (0,0) -- (0.33,0) node[midway,above] {\tt a};
        \draw[thick,red] (0.33,0) -- (0.33+0.618,0) node[midway,above] {\tt b};
        \draw[thick,blue] (0.33+0.618,0) -- (0.33+0.618+0.382,0) node[midway,above] {\tt c};

        \draw[thick,green] (0,-0.7) -- (0.33,-0.7) node[midway,above] {\tt a};
        \draw[thick,blue] (0.33,-0.7) -- (0.712,-0.7) node[midway,above] {\tt c};
        \draw[thick,red] (0.712,-0.7) -- (1.33,-0.7) node[midway,above] {\tt b};

        \draw[->,thick] (-0.01,0) .. controls (-0.08,-0.1) and (-0.08,-0.6) .. (-0.01,-0.7) node[midway,left] {$S_\B$};

        \draw[thick,yellow] (1.53,0) -- (1.86,0) node[midway,above] {\tt d};
        \draw[thick,yellow] (1.53,-0.7) -- (1.86,-0.7) node[midway,above] {\tt d};

        \draw[->,thick] (1.52,0) .. controls (1.45,-0.1) and (1.45,-0.6) .. (1.52,-0.7) node[midway,left] {$S_{\overline{\B}}$};
    \end{tikzpicture}
    \caption{The IET of Example \protect~\ref{ex:merging} after the application of $\sigma_\B$.}
    \label{fig:splitting-merged}
\end{figure}
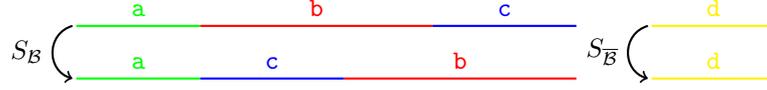
\end{example}

\begin{remark}
\label{rmk:separating}
Each minimal component is a union of sub-intervals of $[\ell, r)$, starting and ending at a $0$-connection, so it can be isolated using a proper sequence of splitting steps.
\end{remark}

\begin{remark}
\label{rmk:arbitrarySplitting}
The choice to define splitting as a two-valued map is purely one of presentation.
One could just as well split off all proper $T$-invariant contiguous sub-alphabets in a single step whenever they appear during the induction process, obtaining three or more IETs when possible.
We chose two in the interest of clarity.
\end{remark}

\subsection{First-return maps to cylinders}
\label{sec:firstreturn}

In~\cite{branching} it is shown that if $T$ is a regular IET, then for every $w \in \LL(T)$ the transformation induced by $T$ on $I_w$ is of the form $\chi(T)$, with
$\chi \in \{ \rho, \lambda \}^*$.
In this case, each morphism corresponds to a non-degenerate step, the regularity hypothesis assuring that no degenerate cases or splitting may happen.

We generalize this result to irregular IETs.
We first recall the following lemmata from the above paper.

\begin{lemma}[\cite{branching}]
\label{lem:allAdmissible}
Let $T$ be a $k$-IET on the interval $I = [\ell, r)$ over the alphabet $\A$.
For every $w \in \LL (T)$, $I_w$ is an admissible interval.
\end{lemma}

\begin{lemma}[\cite{branching}]
\label{lem:extremeInclusion}
Let $T$ be a $k$-IET on the interval $I =[\ell, r)$.
We define the sets
$$
    Z(T) = [\ell, r'),
    \quad \text{and} \quad
    Y(T) = [\ell', r),
$$
where $r'$ is the right-most
(resp., $\ell'$ is the left-most)
element in $D(T) \cup D(T^{-1})$.

If an interval $J$ is strictly included in $I$ and $J$ is admissible for $T$, then either $J \subset Z(T)$ or $J \subset Y(T)$.
\end{lemma}

\begin{lemma}[\cite{branching}]
\label{lem:divPoints}
Let $T$ be a $k$-IET on $I = [\ell, r)$.
Then, for two admissible intervals $J_1 \subset J_2$, $\text{Div} (J_2, T) \subset \text{Div} (J_1, T)$.
\end{lemma}

We employ these lemmata in the following proposition.

\begin{proposition}
\label{pro:inductionChain}
Let $T$ be a $k$-IET on $I = [\ell, r)$ over the alphabet $\A$.
For every $w \in \LL(T)$, there exists a finite sequence of steps
$\chi \in \{ \rho, \lambda, \sigma\}^*$
such that $\chi(T)$ is the first-return map of $T$ to $I_w$. 
\end{proposition}
\begin{proof}
We split the argument into four cases.
\begin{enumerate}
\item
Suppose $T$ is regular.
This case is given by \cite[Theorem 4.3]{branching}.

\item
Suppose $T$ is minimal, but not necessarily regular.
Lemma~\ref{lem:allAdmissible} shows that $I_w \subset [\ell, r)$ is admissible.
Set $S_0 = T$ and $J_0 = [\ell, r)$.
By Lemma~\ref{lem:extremeInclusion}, $I_w$ is either contained in $Z(T)$ or $Y(T)$.
The map $S_1$ obtained from $S_0$ after applying $\rho$ in the former case and $\lambda$ in the latter case is an IET both when the step is non-degenerate, as seen above (see also~\cite{branching}),
and when the step is degenerate, by Propositions~\ref{pro:well-behaved-merging} and~\ref{pro:well-behaved-merging-left}.
Furthermore, it is the first-return map to the domain of $S_1$, denoted by $J_1$.

We proceed inductively for $m > 0$ by defining $S_{m+1}$ and $J_{m+1}$ from $S_m$ on $J_m$ analogously to the base case.
By Lemma~\ref{lem:divPoints}, every $J_m$ containing $I_w$ has endpoints in the finite set $\text{Div}(I_w, T) \cup \{r\}$.
The inclusion chain of the $J_m$ is strictly decreasing among these finitely many super-intervals, and thus terminates at $J_{m'} = I_w$ after finitely many steps.
By construction, each composite step yields a first return on its given sub-interval, and thus the composite map is the first-return map on $I_w$.

\item
Suppose $T$ is a disjoint union of periodic components.
A finite chain of $\sigma$ isolates all components which do not contain $I_w$, since the number of such components is bounded by $k$ for a $k$-IET.
A standard induction covered in the first and second cases follows on this component.

\item
Finally, suppose $T$ is a mix of the above, i.e., $T$ can be decomposed into a non-zero number of periodic components and a non-zero number of minimal components.
By the third point, a finite application of $\sigma$ removes all periodic components that do not contain $I_w$.
A subsequent finite application of (possibly degenerate) $\rho$ and $\lambda$ then isolates $I_w$ in the domain as a consequence of point $2$.
This chain of morphisms is also finite.
\end{enumerate}
\end{proof}

\begin{example}
\label{ex:induction}
Let $T = T_1$ be the IET defined in Example~\ref{ex:non-minimal-IET}.
The induction on the interval $I_c$ begins with an application of $\rho$ in a degenerate case, followed by two successive applications of splitting, $\sigma_{\{ {\tt a} \}}$ and $\sigma_{\{ {\tt b} \}}$, concluding with a (non-degenerate) right Rauzy step $\rho$ on the minimal IET on the interval $I_c \cup I_d$ (see Figure~\ref{fig:ietBigInduce}, where the induced IET on $I_c$ is denoted by $T_7$).

In a similar way, 
we can obtain the induced IETs on $I_a$ and $I_b$
by a right Rauzy step $\rho$ followed by a splitting step $\sigma_{\{ {\tt a} \}}$,
and by a right Rauzy Step $\rho$ followed by two splitting steps $\sigma_{\{ {\tt a} \}}$ and $\sigma_{\{ {\tt b} \}}$,
respectively 
($T_3$ and $T_5$ in the figure).

\begin{figure}[H]
    \centering
    \begin{tikzpicture}[x=6.5cm, y=1cm]
        \draw[thick,green] (0,0) -- (0.1666667,0) node[midway,above] {\tt a};
        \draw[thick,red] (0.1666667,0) -- (0.3333333,0) node[midway,above] {\tt b};
        \draw[thick,blue] (0.3333333,0) -- (0.6262266,0) node[midway,above] {\tt c};
        \draw[thick,yellow] (0.6262266,0) -- (0.8333333,0) node[midway,above] {\tt d};
        \draw[thick,magenta] (0.8333333,0) -- (1,0) node[midway,above] {\tt e};

        \draw[thick,magenta] (0,-0.7) -- (0.1666667,-0.7) node[midway,above] {\tt e};
        \draw[thick,red] (0.1666667,-0.7) -- (0.3333333,-0.7) node[midway,above] {\tt b};
        \draw[thick,yellow] (0.3333333,-0.7) -- (0.5404401,-0.7) node[midway,above] {\tt d};
        \draw[thick,blue] (0.5404401,-0.7) -- (0.8333333,-0.7) node[midway,above] {\tt c};
        \draw[thick,green] (0.8333333,-0.7) -- (1,-0.7) node[midway,above] {\tt a};

        \draw[dotted] (0.1666667,0) -- (0.1666667,-0.7) {};
        \draw[dotted] (0.3333333,0) -- (0.3333333,-0.7) {};

        \draw[->, thick] (-.05, 0) .. controls (-0.1, 0.1) and (-0.1, -0.8) .. (-.05, -0.7) node[midway, left] {$T_1$};

        \draw[thick,green] (0,-1.6) -- (0.1666667,-1.6) node[midway,above] {\tt a};
        \draw[thick,red] (0.1666667,-1.6) -- (0.3333333,-1.6) node[midway,above] {\tt b};
        \draw[thick,blue] (0.3333333,-1.6) -- (0.6262266,-1.6) node[midway,above] {\tt c};
        \draw[thick,yellow] (0.6262266,-1.6) -- (0.8333333,-1.6) node[midway,above] {\tt d};

        \draw[thick,green] (0,-2.3) -- (0.1666667,-2.3) node[midway,above] {\tt a};
        \draw[thick,red] (0.1666667,-2.3) -- (0.3333333,-2.3) node[midway,above] {\tt b};
        \draw[thick,yellow] (0.3333333,-2.3) -- (0.5404401,-2.3) node[midway,above] {\tt d};
        \draw[thick,blue] (0.5404401,-2.3) -- (0.8333333,-2.3) node[midway,above] {\tt c};

        \draw[dotted] (0.1666667,-1.6) -- (0.1666667,-2.3) {};
        \draw[dotted] (0.3333333,-1.6) -- (0.3333333,-2.3) {};

        \draw[->, thick] (-.05, -1.6) .. controls (-0.1, -1.5) and (-0.1, -2.4) .. (-.05, -2.3) node[midway, left] {$T_2$};

        \draw[thick,green] (0,-3.2) -- (0.1666667,-3.2) node[midway,above] {\tt a};

        \draw[thick,green] (0,-3.9) -- (0.1666667,-3.9) node[midway,above] {\tt a};

        \draw[->, thick] (-.05, -3.2) .. controls (-0.1, -3.1) and (-0.1, -4.0) .. (-.05, -3.9) node[midway, left] {$T_3$};

        \draw[thick,red] (0.3,-3.2) -- (0.4666667,-3.2) node[midway,above] {\tt b};
        \draw[thick,blue] (0.4666667,-3.2) -- (0.7595599,-3.2) node[midway,above] {\tt c};
        \draw[thick,yellow] (0.7595599,-3.2) -- (0.9666667,-3.2) node[midway,above] {\tt d};

        \draw[thick,red] (0.3,-3.9) -- (0.4666667,-3.9) node[midway,above] {\tt b};
        \draw[thick,yellow] (0.4666667,-3.9) -- (0.6737735,-3.9) node[midway,above] {\tt d};
        \draw[thick,blue] (0.6737735,-3.9) -- (0.9666667,-3.9) node[midway,above] {\tt c};

        \draw[dotted] (0.4666667,-3.2) -- (0.4666667,-3.9) {};

        \draw[->, thick] (0.29, -3.2) .. controls (0.24, -3.1) and (0.24, -4.0) .. (0.29, -3.9) node[midway, left] {$T_4$};

        \draw[thick,green] (0,-4.8) -- (0.1666667,-4.8) node[midway,above] {\tt a};
        
        \draw[thick,green] (0,-5.5) -- (0.1666667,-5.5) node[midway,above] {\tt a};

        \draw[->, thick] (-.05, -4.8) .. controls (-0.1, -4.7) and (-0.1, -5.6) .. (-.05, -5.5) node[midway, left] {$T_3$};

        \draw[thick,red] (0.3,-4.8) -- (0.4666667,-4.8) node[midway,above] {\tt b};

        \draw[thick,red] (0.3,-5.5) -- (0.4666667,-5.5) node[midway,above] {\tt b};

        \draw[->, thick] (0.29, -4.8) .. controls (0.24, -4.7) and (0.24, -5.6) .. (0.29, -5.5) node[midway, left] {$T_5$};

        \draw[thick,blue] (0.6,-4.8) -- (0.8928932,-4.8) node[midway,above] {\tt c};
        \draw[thick,yellow] (0.8928932,-4.8) -- (1.1,-4.8) node[midway,above] {\tt d};

        \draw[thick,yellow] (0.6,-5.5) -- (0.8071068,-5.5) node[midway,above] {\tt d};
        \draw[thick,blue] (0.8071068,-5.5) -- (1.1,-5.5) node[midway,above] {\tt c};

        \draw[->, thick] (0.59, -4.8) .. controls (0.54, -4.7) and (0.54, -5.6) .. (0.59, -5.5) node[midway, left] {$T_6$};

        \draw[thick,green] (0,-6.4) -- (0.1666667,-6.4) node[midway,above] {\tt a};

        \draw[thick,green] (0,-7.1) -- (0.1666667,-7.1) node[midway,above] {\tt a};

        \draw[->, thick] (-.05, -6.4) .. controls (-0.1, -6.3) and (-0.1, -7.2) .. (-.05, -7.1) node[midway, left] {$T_3$};

        \draw[thick,red] (0.3,-6.4) -- (0.4666667,-6.4) node[midway,above] {\tt b};

        \draw[thick,red] (0.3,-7.1) -- (0.4666667,-7.1) node[midway,above] {\tt b};

        \draw[->, thick] (0.29, -6.4) .. controls (0.24, -6.3) and (0.24, -7.2) .. (0.29, -7.1) node[midway, left] {$T_5$};

        \draw[thick,blue] (0.6,-6.4) -- (0.6857864,-6.4) node[midway,above] {\tt c};
        \draw[thick,yellow] (0.6857864,-6.4) -- (0.8928932,-6.4) node[midway,above] {\tt d};

        \draw[thick,yellow] (0.6,-7.1) -- (0.8071068,-7.1) node[midway,above] {\tt d};
        \draw[thick,blue] (0.8071068,-7.1) -- (0.8928932,-7.1) node[midway,above] {\tt c};

        \draw[->, thick] (0.59, -6.4) .. controls (0.54, -6.3) and (0.54, -7.2) .. (0.59, -7.1) node[midway, left] {$T_7$};

        \draw[->, thick] (1.05, -0.7) .. controls (1.1, -0.8) and (1.1, -1.5) .. (1.05, -1.6) node[midway, right] {$\rho: T_1 \mapsto T_2$};
        \draw[->, thick] (1.05, -2.3) .. controls (1.1, -2.4) and (1.1, -3.1) .. (1.05, -3.2) node[midway, right] {$\sigma_{\{a\}}: T_2 \mapsto (T_3, T_4 )$};
        \draw[->, thick] (1.15, -3.9) .. controls (1.2, -4.0) and (1.2, -4.7) .. (1.15, -4.8) node[midway, right] {$\sigma_{\{b\}}: T_4 \mapsto (T_5, T_6 )$};
        \draw[->, thick] (1.15, -5.5) .. controls (1.2, -5.6) and (1.2, -6.3) .. (1.15, -6.4) node[midway, right] {$\rho: T_6 \mapsto T_7$};
    \end{tikzpicture}
    \caption{Extended branching Rauzy induction on the interval $I_c$ of a non-minimal $5$-IET.}
    \label{fig:ietBigInduce}
\end{figure}
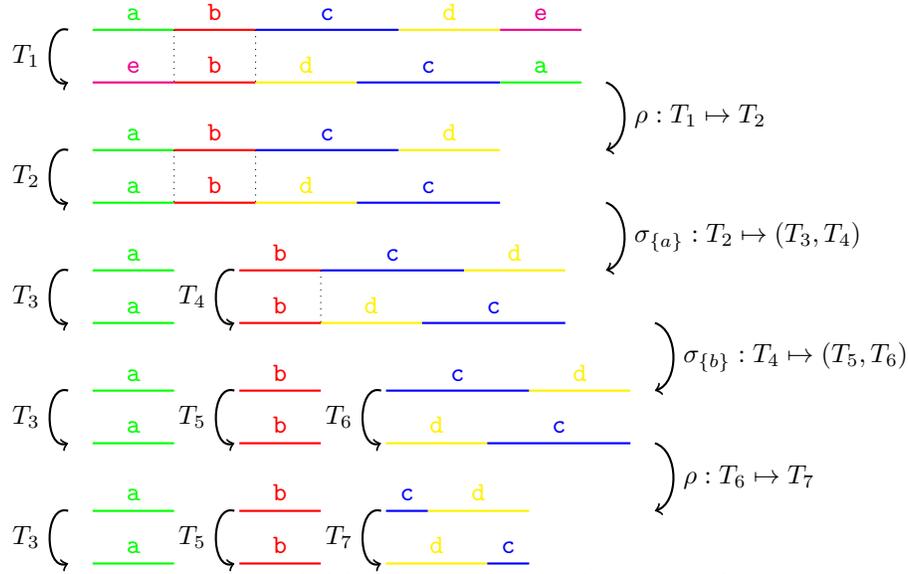
\end{example}

\section{Extension Graphs and Alsinicity}
\label{sec:forest}

In order to generalize the clustering result of ~\cite{DolceHughes2025} to all standard interval exchange transformations, and to specify the specific clustering permutation, we require order conditions and extension graphs.

Let $\LL \subset \A^*$ be a language.
The \emph{extension graph} $\G_\LL(w)$ of a word $w \in \LL$ is the undirected bipartite graph with vertices the disjoint union of
$$
    L_\LL(w) = \{a \in \A \; : \; aw \in \LL \}
    \quad \text{and} \quad
    R_\LL(w) = \{b \in \A \; : \; wb \in \LL \},
$$
and edges
$$
    B_\LL(w) = \{ (a,b) \in \A^2 \; : \; awb \in \LL \}.
$$
A word $w \in L$ such that $| L_\LL(w)| > 1$ (resp., ($|R_\LL(w)| > 1$) is said to be \emph{left-special} (resp., \emph{right-special}).
A word that is both left- and right-special is called \emph{bispecial}.

An extension graph $\G_\LL(w)$ is \emph{compatible} with two orders $<_1$ and $<_2$ on $\A$ if
$$
    \forall
    (a,b), (c,d) \in B_\LL(w),
    \qquad
    a <_1 c \Rightarrow b \le_2 d.
$$
When $\LL$ is clear from the context, we just write $\G(w)$, $L(w)$, $R(w)$ and $B(w)$.

A language $\LL$ is said to be \emph{dendric} if the extension graph of every $w \in \LL$ is acyclic and connected (whence the original name \emph{tree set} in~\cite{acyclicconnectedtree}).
A language \emph{alsinic} if the extension graph of every word in it is a forest (see~\cite{DolceHughes2025}).
Since every forest is a planar graph, for every word $w$ in an alsinic language is it possible to find two order $<_{1,w}$ and $<_{2,w}$ such that $\G(w)$ is compatible with them.

A language $\LL$ is \emph{ordered dendric} (resp., \emph{ordered alsinic}) for two orders $<_1$ and $<_2$ if every $\G(w)$, with $w \in \LL$, is compatible for $<_1$ and $<_2$ (in~\cite{bifixcodesIETs} the term \emph{planar tree} was used since the edges do not cross when drawing the vertices of $L(w)$ on the left ordered by $<_1$ and the vertices of $R(w)$ on the right ordered by $<_2$).

An example of a family of ordered dendric languages is given by Sturmian languages, i.e., coding of an irrational rotation on a circle.
On the other hand, Arnoux-Rauzy words on more than two letters are dendric but not ordered dendric~\cite{bifixcodesIETs}.

Ordered alsinic languages are strictly linked to IETs.

\begin{theorem}[\cite{FerencziZamboni08,FerencziHubertZamboni24}]
Let $\LL \subset \A^*$.
\begin{itemize}
\item $\LL$ is the language of an IET if and only if it is a recurrent alsinic language.

\item $\LL$ is the language of a minimal IET if and only if it is an aperiodic, uniformly recurrent alsinic language.

\item $\LL$ is the language of a regular IET if and only if it is a recurrent ordered dendric set.
\end{itemize}
\end{theorem}

Note that in~\cite{FerencziZamboni08,FerencziHubertZamboni24} the hypothesis of uniform recurrence was used for the last point.
However, as proved in~\cite{DolcePerrin17}, recurrent dendric sets are also uniformly recurrent. 
As observed earlier, clustering words are associated with DIETs. Recall also that a DIET can be viewed as an IET having sub-intervals of integral (or rational) length.
We make this link explicit via a reformulation of~\cite[Theorem 8]{FerencziHubertZamboni23}.

For two total orders $<_1$ and $<_2$ on alphabet $\A$, the order condition of Ferenczi, Hubert, and Zamboni in~\cite{FerencziHubertZamboni23} states that, for a factor $w$, whenever $awb$ and $a'wb'$ occur with $a <_2 a'$ and $b <_1 b'$, the edges $\{a, b\}$ and $\{a', b'\}$ do not cross.
Hence, a language is ordered alsinic for two orders exactly when it satisfies the order condition for those two orders.
In the same paper, the authors show that a primitive word $w$ is $\pi$-clustering iff every bispecial factor in $\LL (w^{\omega})$ satisfies the order condition with respect to $<_\pi$ and $<_\A$, where
$<_{\A}$ is the order on the alphabet $\A$ and $<_\pi$ the order given by $a <_\pi b$ when $\pi^{-1}(a) <_\A \pi^{-1}(b)$.

If $u \in \LL (w^{\omega})$ is a non-bispecial factor , i.e., if $|L(u)| \le 1$ or $|R(u)| \le 1$,
the extension graph $\G(u)$ is trivially a tree compatible with any pair of orders.

Thus, the only factors in this language that could possibly fail the order condition are the bispecial factors.
Replacing the order condition by alsinicity, we obtain the following reformulation of a result in~\ref{thm:FHZ23}

\begin{theorem}[\cite{FerencziHubertZamboni23}]
\label{thm:FHZ23}
A word $w \in \A^*$ is $\pi$-clustering if and only if $\LL(w^{\omega})$ is ordered alsinic with respect to $<_\pi$ and $<_\A$.
\end{theorem}

\begin{example}
\label{ex:extensiongraph}
Let $w$ and $T$ be as in Example~\ref{ex:ebwt-DIET}.
Then
$$
    {\tt c} <_\pi {\tt b} <_\pi {\tt a}
    \quad \text{and} \quad
    {\tt a} <_{\A} {\tt b} <_{A} {\tt c}.
$$
The extension graphs of the empty word and the letter ${\tt a}$ are shown in Figure~\ref{fig:extensiongraph} (left and center).
It is straightforward to show that $\G(w)$ contains only one edge for every $w \in \LL \setminus \{ \varepsilon, {\tt a} \}$ (e.g., right of Figure~\ref{fig:extensiongraph}).

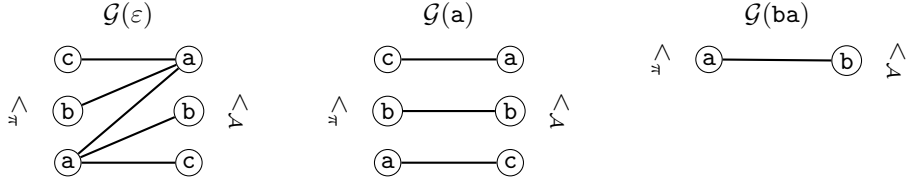
\begin{figure}[H]
    \centering
    \tikzset{node/.style={draw, circle, inner sep=0.4mm}}
    \tikzset{title/.style={minimum size=0.5cm,inner sep=1pt}}
    \begin{tikzpicture}
        \node[title] (Ee) {$\G(\varepsilon)$};
        \node[node](ecl) [below left = 0.2cm and 0.3cm of Ee] {\tt c};
        \node[node](ebl) [below = 0.3cm of ecl] {\tt b};
        \node[node](eal) [below = 0.3cm of ebl] {\tt a};
        \node[title](ePi) [left = 0.2cm of ebl] {\eqmathbox{\rotatebox[origin=c]{-90}{$<_\pi$}}};
        \node[node](ear) [below right = 0.2cm and 0.3cm of Ee] {\tt a};
        \node[node](ebr) [below = 0.3cm of ear] {\tt b};
        \node[node](ecr) [below = 0.3cm of ebr] {\tt c};
        \node[title](eA) [right = 0.2cm of ebr] {\eqmathbox{\rotatebox[origin=c]{-90}{$<_\A$}}};
        \path
            (ecl) edge[thick] (ear)
            (ebl) edge[thick] (ear)
            (eal) edge[thick] (ear)
            (eal) edge[thick] (ebr)
            (eal) edge[thick] (ecr)
        ;

        \node[title] (Ea)[right = 3.5cm of Ee] {$\G({\tt a})$};
        \node[node](acl) [below left = 0.2cm and 0.3cm of Ea] {\tt c};
        \node[node](abl) [below = 0.3cm of acl] {\tt b};
        \node[node](aal) [below = 0.3cm of abl] {\tt a};
        \node[title](aPi) [left = 0.2cm of abl] {\eqmathbox{\rotatebox[origin=c]{-90}{$<_\pi$}}};
        \node[node](aar) [below right = 0.2cm and 0.3cm of Ea] {\tt a};
        \node[node](abr) [below = 0.3cm of aar] {\tt b};
        \node[node](acr) [below = 0.3cm of abr] {\tt c};
        \node[title](aA) [right=0.2cm of abr] {\eqmathbox{\rotatebox[origin=c]{-90}{$<_\A$}}};
        \path
            (acl) edge[thick] (aar)
            (abl) edge[thick] (abr)
            (aal) edge[thick] (acr)
        ;

        \node[title] (Eba)[right = 3.5cm of Ea] {$\G({\tt ba})$};
        \node[node](baal) [below left = 0.2cm and 0.3cm of Eba] {\tt a};
        \node[title](baPi) [left = 0.2cm of baal] {\eqmathbox{\rotatebox[origin=c]{-90}{$<_\pi$}}};
        \node[node](babr) [below right = 0.2cm and 0.3cm of Eba] {\tt b};
        \node[title](aA) [right = 0.2cm of babr] {\eqmathbox{\rotatebox[origin=c]{-90}{$<_\A$}}};
        \path
            (baal) edge[thick] (babr)
        ;
    \end{tikzpicture}
    \caption{The extension graphs of $\varepsilon, {\tt a}$ and ${\tt ba}$ in $\LL(T)$.}
    \label{fig:extensiongraph}
\end{figure}
\end{example}

Following the same argument seen in Section~\ref{sec:iets}, Theorem~\ref{thm:FHZ23} can be generalized to multiset of words.

\begin{proposition}[\cite{DolceHughes2025}]
A multiset $W \subset \A^*$ is $\pi$-clustering if and only if every for every $w \in W$, the language $\LL(w^\omega)$
is ordered alsinc with respect to $<_{\pi_w}$ and $<_{\A_w}$,
where $\A_w = \A \cap \LL(w)$ and $\pi_w$ is the restriction of $\pi$ to $\A_w$.
\end{proposition}

The next result follows from the proof in~\cite[Lemma 2]{enumerationneutral}.

\begin{lemma}
\label{lem:regionconnected}
Let $T$ be an IET.
Two letters $a,b \in L(\varepsilon)$ are in the same left connected component
(resp., $a,b \in R(\varepsilon)$ are in the same right connected component)
of $\G(\varepsilon)$
if and only if
$T(I_a)$ and $T(I_b)$
(resp., $I_a$ and $I_b$)
are in the same region.
\end{lemma}

\section{Return Words in IET Languages}
\label{sec:piclusteringsection}

As an application of this induction procedure, we give a proof of the following theorem.
The same result was recently showed, independently, in~\cite{ferenczi2025clustering}, using different tools.

\begin{theorem}
\label{thm:picluster}
Let $T$ be an interval exchange transformation over the alphabet $\A$ with a permutation $\pi \in S_\A$.
All return words in the language of $T$ are $\pi$-clustering.
\end{theorem}

Before handling the general case, in which we allow our IET to have irregular components, we show the following lemma.

\begin{lemma}
\label{lem:clustering1}
Let $T$ be a regular $n$-IET, $n \ge 1$, with associated permutation $\pi$.
Let $\LL(T)$ be its language.
All return words $w \in \LL(T)$ are $\pi$-clustering.
\end{lemma}
\begin{proof}
Let $u \in \LL(T)$ and $w \in \R(u) \subset \LL(T)$.
We handle the following cases for $v \in \LL(w^{\omega})$ using Theorem~\ref{thm:FHZ23}.
We denote the length $|w|$ of the word $w$ by $r$.
Note that, since $T$ is regular, $\LL(T)$ is uniformly recurrent ordered alsinic for the orders $<_{\A}$ and $<_{\pi}$.

First, suppose $|u| \geq |w|$.
\begin{enumerate}
\item
Let $|v| \ge |w|$.
Then $v = swp$ for a suffix $s$ of $w$ and a prefix $p$ of $w$.
In the language $\LL(w^{\omega})$, the word $v$ can be extended on the left and on the right by a single letter in a unique way. Consequently, $v$ is neither left-special nor right-special in $\LL(w^{\omega})$.
As a result, the extension graph of $v$ trivially satisfies the ordered alsinic condition, since it only contains one edge between two vertices.

\item
Now suppose that $|v| < |w|$.
Then, either $v$ straddles two occurrences of $w$, or $v$ is a factor of $w$.
The second case is trivial, since $v$ rests in $\LL(T)$ by definition.
In the former case, $v \in \LL(w^2) \setminus \LL(w) \subset \LL(uw) \subset \LL(T)$.
Thus, $v \in \LL(T)$, and the extension graphs of all factors of $\LL(T)$ are compatible w.r.t. the orders $<_A$ and $<_{\pi}$ as a consequence of the regularity of $T$.
Hence, the extension graph of $v$ is ordered alsinic with respect to $<_A$ and $<_{\pi}$.
\end{enumerate}

Now suppose $|u| < |w|$.
\begin{enumerate}
\item
Let $|v| \geq |w|$.
The argument is identical to that in the first case with $|u| \ge |w|$ and $|v| \geq |w|$, and the ordered alsinicity of the extension graph of $v$ with respect to the orders $<_\A$ and $<_{\pi}$ follows.

\item
Finally, suppose $|v| < |w|$.
If $v \in \LL(w) \cup \LL(u) \subset \LL$, then we immediately obtain the compatibility of the extension graph of $v$ from the ordered alsinicity of $\LL$.
Otherwise, if $v \in \LL(ww) \setminus \LL(w)$, then we can write $v = sup$.
In such a case, by a similar argument as above, the word $v$ can be uniquely extended on the left by one letter and uniquely extended on the right by one letter.
Consequently, we obtain the ordered alsinicity of the extension graph of $v$ with respect to the desired orders via a symmetrical argument to that of the first case.
\end{enumerate}

As a consequence of the above cases and Theorem~\ref{thm:FHZ23}, all return words in the language of a regular IET with associated permutation $\pi$ are $\pi$-clustering in the Burrows-Wheeler sense.
\end{proof}

\begin{remark}
\label{rem:regToOrdAls}
Note that in the proof of Lemma~\ref{lem:clustering1}, the only property of $T$ used is that $\LL(T)$ is ordered alsinic with respect to $<_{\A}$ and $<_{\pi}$.
In particular, the statement of Lemma~\ref{lem:clustering1} remains valid if one replaces "regular IET" with "IET whose language is ordered alsinic with respect to $<_{\A}$ and $<_{\pi}$".
\end{remark}

Let us recall the morphisms
$\alpha_{a,b}, \tilde{\alpha}_{a,b} :\A^* \to \A^*$
defined in Section~\ref{sec:morphisms} as
$$
    \alpha_{a,b} =
    \begin{cases}
        a \mapsto a b \\
        c \mapsto c, & c \neq a
    \end{cases}
    \qquad \text{and} \qquad
    \tilde{\alpha}_{a,b} =
    \begin{cases}
        a \mapsto b a \\
        c \mapsto c, & c \neq a
    \end{cases},
$$
with $a,b \in \A$,
as well as the inclusions
$\iota_{\B, \A} : \B^* \rightarrow \A^*$
and
$\iota_{\overline{\B}, \A} :\overline{\B}^* \rightarrow \A^*$
defined in Section~\ref{sec:splitting}, with $\B \subset \A$.

We use these morphisms to step back from $I_w$ to $[\ell, r)$.
Recall that these morphims are $\alpha_{a,b}: a \mapsto ab$ and $\tilde{\alpha}_{a,b}: a \mapsto ba$, which fix all other letters, as well as the inclusion $\iota_{\A, \B}: \A^* \rightarrow (\A \sqcup \B)^*$.
We begin by showing that primitivity is preserved by these morphisms.

\begin{lemma}
\label{lem:primitive}
Let $w$ be a primitive word over an ordered alphabet $\A$.
Let $a \in \A$ and $b \in \A \cup \B$, with $\B$ an alphabet disjoint from $\A$.
For every
$\varphi \in S_\A \cup \{ \alpha_{a, b}, \tilde{\alpha}_{a,b}, \iota_{\A,\B} \}$ the word $\varphi(w)$ is primitive.
\end{lemma}
\begin{proof}
Let $u = \varphi(w)$.
\begin{enumerate} 
\item
We first consider the case $\varphi \in S_\A$.
Let $u = v^d$ for some primitive word $v = a_1 \ldots a_m \in \A^*$.
Set $v' = \varphi^{-1}(a_1) \cdots \varphi^{-1}(a_m)$.
After applying the inverse permutation $\varphi^{-1} \in S_\A$, we obtain $w = v'^{d}$.
By the primitivity of $w$, we find that $d=1$.

\item Suppose now that $\varphi = \alpha_{a, b}$.
We first consider the case where $b \in \A$.
Let $u = v^d$ for some primitive word $v \in \A^*$.
We claim that $v$ cannot end with the letter $a$, since otherwise $u$ would also end with $a$, which is impossible, since by construction every $a$ in $u = \alpha_{a,b}(w)$ is followed by the letter $b$.
Therefore, $v = v_1 ab v_2 ab \cdots v_m c$, where $c \in \A \setminus \{ a \}$ and $v_i \in (\A \setminus \{ a \})^{*}$ for every $i$.
Let $v' = v_1 a v_2 \cdots a v_m$.
We have $v = \alpha_{a,b}(v')$.
Consequently, $w = v'^d$, and thus, by primitivity of $w$, $d=1$.

Now suppose that $b \notin \A$.
Let $u = v^d$ for some primitive word $v \in \A^*$.
By the same argument as above, we can write
$v = v_1 ab v_2 ab \cdots v_m c$ with $c \in \A \setminus \{ a \}$ and $v_i \in ((\A \setminus \{a \}) \cup \{ b \})^*$.
Since in $u = \alpha_{a,b}(w)$ every occurrence of $b$ is, by construction, preceded by the letter $a$, we have $v_i \in (\A \setminus \{a \})^*$.
We then conclude as above.

\item
The case $u = \tilde{\alpha}_{a,b}(w)$ is proved in a symmetric way as in point (2).

\item
The case $u = \iota_{\A,\B}$ is immediate.
\end{enumerate}
\end{proof}

We now show that the property of being clustered is preserved by the morphisms defined in Section ~\ref{sec:Rauzy}.
\begin{lemma}
\label{lem:clustering}
Let $w$ be a primitive word over $\A = \{a_1 < a_2 < \ldots < a_k \}$.
Suppose $w$ is $\pi$-clustering on $\A$ for some permutation $\pi \in S_\A$.
\begin{enumerate}
\item If $\mu \in S_\A$, then $\mu(w) \in \A'^*$ is $\pi'$-clustering, with $\pi' \in S_{\A'}$ and $\A' = \{ \mu(a_1) < \ldots < \mu(a_k) \}$.

\item If $b = a_1$, $\pi^{-1}(a) = a_i$, and $\pi^{-1}(b) = a_{i+1}$ for a certain $1 \le i < k$, then $\alpha_{a,b}(w) \in \A^*$ is $\pi'$-clustering, with $\pi' \in S_\A$.

\item If $b = a_k$, $\pi^{-1}(b) = a_i$ and $\pi^{-1}(a) = a_{i+1}$ for a certain $1 \le i < k$, then $\alpha_{a,b}(w) \in \A^*$ is $\pi'$-clustering, with $\pi' \in S_\A$.

\item If $\pi^{-1}(b) = a_1$, $a=a_i$ and $b=a_{i+1}
$ for a certain $1 \le i < k$, then $\tilde{\alpha}_{a,b}(w) \in \A'^*$ is $\pi'$-clustering, with $\A' = \{a_i < a_1 < \ldots < a_{i-1} < a_{i+1} < \ldots < a_k \}$ and $\pi' \in S_{\A'}$.

\item If $\pi^{-1}(b) = a_k$, $b=a_i$ and $a=a_{i+1}$ for a certain $1 \le i < k$, then $\tilde{\alpha}_{a,b}(w) \in \A'^*$ is $\pi'$-clustering, with $\A' = \{a_1 < \ldots < a_i < a_{i+2} < \ldots < a_k < a_{i+1} \}$ and $\pi' \in S_{\A'}$.

\item If $b \notin \A$, then $\alpha_{a,b}(w) \in \A'^*$ is $\pi'$-clustering, with $\A' = \{ a_1 < \ldots < a_k < b \}$ and $\pi' \in S_{\A'}$;
is it also $\pi''$-clustering with $\pi'' \in S_{\A''}$ when considered as word over $\A'' = \{ b < a_1 < \ldots < a_k \}$.

\item If $\B = \{ b_1 < \ldots < b_{h} \}$ is disjoint from $\A$, then $\iota_{\A,\B}(w) \in \A'^*$ is $\pi'$-clustering, with $\A' = \{ a_1 < \ldots < a_k < b_1 < \ldots < b_{h} \}$ and $\pi' \in S_{\A'}$.
\end{enumerate}
\end{lemma}
\begin{proof}
We proceed by addressing each case separately.
\begin{enumerate}
\item
The permutation $\mu$ sends each letter $a_i$ to $\mu(a_i)$.
Since $w$ is $\pi$-clustering, $\bwt{\A}{w}$ consists of contiguous runs - possibly of length $0$ - of each letter of $\A$.
Since an application of $\mu$ to $w$ simply amounts to letter renaming, we immediately obtain clustering of $\mu(w)$.
Defining $\pi' = \mu \circ \pi \circ \mu^{-1} \in S_{\A'}$, via elementary permutation properties we see that $w$ is $\pi'$-clustering.

\item
Define $\pi' \in S_{\A}$ by modifying $\pi$ such that $\pi^{-1}(a)$ and $\pi^{-1}(b)$ are adjacent in the cycle.
If $\pi^{-1}(a) = a_i$ and $\pi^{-1}(b) = a_{i+1}$, then the runs $\pi(a)$ and $\pi(b)$ appear consecutively in $\bwt{\A}{w}$.
Replacing each $a$ by $ab$ in $w$ connects the runs $\pi(a)$ and $\pi(b)$ into a single-run adjacency in $\bwt{\A}{\alpha_{a,b} (w)}$.
Thus, $\alpha_{a,b}(w)$ is $\pi'$-clustering.

\item
This proof is symmetrical to that of the case (2).

\item
Let $a = a_i$ and $b = a_{i+1}$, $1 \le i \le k$.
The letter $b$ is the first letter of the permutation $\pi$, so $b = a_{\pi(1)}$ is the first letter in $\bwt{\A}{w}$.
The morphism $\tilde{\alpha}_{a,b}$ sends $a$ to $ba$, so each new $b$ is created inside the existing run of $a$'s directly preceding it.
We now reorder the alphabet, obtaining
$\A' = \{a_i < a_1 < \ldots < a_{i-1} < a_{i+1} < \ldots < a_k\}$,
such that $a$ becomes the smallest letter in the ordering of $\A'$.
With this order the runs appear consecutively, and we obtain $\pi' \in S_{\A'}$ by cyclically rotating $\pi$ in the same fashion.
Thus, $\tilde{\alpha}_{a,b}$ is $\pi'$-clustering.

\item
This case is symmetric to that case (4).
In this case, $b$ is the letter appearing in the last run of the $\bwt{\A}{w}$.
A reordered alphabet
$\A' = \{a_1 < \ldots < a_i <a_{i+2} < \ldots < a_k < a_{i+1} \}$ 
moves $a$ to the end, resulting in a permutation $\pi'$ that keeps the two new adjacent runs together.
Clustering follows.

\item
Let $b$ be the largest letter of the new alphabet
$\A' = \A \cup \{ b \}$
with
$\pi' \in S_{\A'}$ defined by inserting $b$ immediately after $a$ in the cyclic order of $\pi$.
After replacing every $a$ with $ab$, all new $b$'s sit inside the run of $a_{\pi(i)} = a$ and nowhere else, so $\bwt{\A}{w}$ consists of contiguous runs once the run of $b$'s is appended after that of $a$.
A similar procedure applies should we choose to append $b$ at the beginning of $\A$ instead of at the end.
Thus, $\alpha_{a,b}$ is $\pi'$-clustering on both extended alphabets.

\item
We note that $\iota_{\A,\B} (w) = w$.
No new letters appear in $w$, so $\bwt{\A \cup \B}{w} = \bwt{\A}{w}$ with zero occurrences of each $b \in \B$.
We define $\pi' \in S_{\A \cup \B}$ by letting $\pi'$ extend $\pi$ and act as the identity on each $b \in \B$.
This keeps the clustering unaltered, hence $w$ is $\pi'$-clustering over $\A \cup \B$.
\end{enumerate}
\end{proof}

We are now able to prove Theorem~\ref{thm:picluster}, providing an application of this machinery.
With an abuse of notation, let us denote permutations of a sub-alphabet of $\A$ as belonging to $S_\A$.

\begin{proof}[of Theorem \ref{thm:picluster}]
Let $T$ be a $k$-IET on an interval $I$ associated with a permutation $\pi \in S_\A$.
Let $w \in \LL(T)$ and $u$ be a return word to $w \in \LL(T)$.
By Proposition~\ref{pro:inductionChain}, there exists $\chi \in \left( S_\A \cup \{ \rho, \lambda, \sigma \} \right)^*$ such that the transformation induced by $T$ to $I_w$ is $\chi(T)$.
Applying Lemma~\ref{lem:clustering}, we obtain that all return words of $\LL (T)$ are $\pi'$-clustering, with $\pi' \in S_\A$. 
By Remark ~\ref{rem:regToOrdAls}, we apply Lemma ~\ref{lem:clustering1} with "regular" replaced by "ordered alsinic" to find that $\pi'$ and $\pi$ coincide.
\end{proof}

\begin{example}
\label{ex:return}
Let $T = T_1$ be the IET defined in Example~\ref{ex:non-minimal-IET}.
As seen in Example~\ref{ex:coding-non-minimal-IET}, its language is
$
    \LL(T) =
    \LL(({\tt ae})^\omega) \cup \LL(({\tt d})^\omega) \cup \LL(\ff),
$
with
$\ff = {\tt bcbbcb}\cdots$
the Fibonacci word.
It is easy to check that
$\R({\tt a}) = \{ {\tt e} \}$
and
$\R({\tt d}) = \{ {\tt d} \}$,
while
$\R({\tt c}) = \{ {\tt c, cd} \}$.

As shown in Example~\ref{ex:induction}, the transformation induced by $T_1$ on $I_{\tt c}$ is obtained via two right Rauzy steps interspersed with two splitting steps.
For the steps
$$
    \rho,
    \quad
    \sigma_{\{\tt a\}},
    \quad
    \sigma_{\{\tt b\}},
    \quad
    \rho
$$
in the induction, we obtain the respective morphisms
$$
    \alpha_{{\tt a}, {\tt e}},
    \qquad
    \iota_{\{ {\tt b,c,d} \}, \{ {\tt a}\}},
    \qquad
    \iota_{\{ {\tt c,d} \}, \{ {\tt b}\}},
    \qquad
    \tilde{\alpha}_{d,c }.
$$
The alphabet of the original IET induced on $I_{\tt c}$ is $\{ {\tt c, d}\}$.
The corresponding return words to $I_{\tt c}$ in the language of the original IET are obtained via the application of the composition $\chi_{\tt c} = \alpha_{{\tt a}, {\tt e}} \circ \iota_{\{ {\tt b,c,d} \}, \{ {\tt a}\}} \circ \iota_{\{ {\tt c,d} \}, \{ {\tt b}\}} \circ \tilde{\alpha}_{d,c}$ to
the subalphabet $\{ {\tt c}, {\tt d} \}$.
We have
$$
    \chi_{\tt c} :
    \begin{cases}
        {\tt d} \mapsto {\tt cd} \\
        {\tt c} \mapsto {\tt c}
    \end{cases}.
$$
Similarly, the transformations induced by $T$ on $I_{\tt a}$ and on $I_{\tt b}$ are
$\sigma_{\{ {\tt a} \}} \circ \rho$
and
$\sigma_{\{ {\tt b} \}} \circ \sigma_{\{ {\tt a} \}} \circ \rho$ respectively.
Applying the associated morphism compositions $\chi_a$ and $\chi_b$ to the subalphabets $\{ {\tt a} \}$ and $\{ {\tt b \}}$ respectively, we obtain
$$
    \chi_{\tt a}:
    {\tt a} \mapsto {\tt ae}
    \quad \text{and} \quad
    \chi_{\tt b}:
    {\tt b} \mapsto {\tt b}.
$$
Note that not all of these steps were necessary to reach induced IETs on the desired intervals.
For example, one could obtain the induced IET on $I_b$ via a single application of $\sigma_{\{ {\tt b} \}}$.
\end{example}

In~\cite{lapointe2021perfectly}, Lapointe asked the following question: "Are return words of a symmetric interval exchange transformation perfectly clustering words?"

Using Theorem~\ref{thm:picluster} we can provide a positive answer.

\begin{corollary}
\label{cor:solutionToLapointe}
Return words in the language of a symmetric IET are perfectly clustering.
\end{corollary}
\begin{proof}
Let $T$ be an IET with associated permutation the symmetric permutation.
Applying Theorem~\ref{thm:picluster} to $\LL(T)$, we obtain that all its return words are clustering with respect to the symmetric permutation, and thus perfectly clustering.
\end{proof}

\section{Future Work}
\label{sec:future}

We briefly indicate two directions suggested by this work.
A first goal concerns local commutation of the morphisms in extended branching Rauzy induction.
Given the multiple choices of morphisms at each induction step, it would be valuable to identify a small set of local commutations (e.g., splitting on a periodic component, then a left Rauzy step on a minimal component is the same as doing so in the opposite order) that aid in understanding when segments of induction paths commute.
A natural problem is to identify other such local relations in the monoid generated by $\{ \lambda, \rho, \sigma \}$.

A second direction is to explicitly extend the construction of Gjerde and Johansen in~\cite{gjerde2002bratteli} beyond the minimal case.
In particular, the authors of this work constructed a partition in towers corresponding to the Cantor version of any minimal standard interval exchange transformation.
The incidence matrices of the corresponding Bratteli diagram were products of morphism matrices corresponding to individual right Rauzy steps (see also~\cite{durand2022dimension} for a review of the results of~\cite{gjerde2002bratteli} from this perspective).
However, traditional right Rauzy induction is not applicable to general standard IETs.
Each merging and splitting step induces a matrix, and thus extended branching Rauzy induction could enable the explicit construction of Bratteli diagrams (and the corresponding Bratteli-Vershik models) from the Cantor system with multiple minimal components arising from an IET.
This would also provide a convenient method of recovering the associated dimension groups.

\section*{Acknowledgments}

This work was partially supported by the Výlet summer program of the Faculty of Information Technology, Czech Technical University in Prague.

\addcontentsline{toc}{section}{References}
\bibliographystyle{plain}
\bibliography{references}

\end{document}